\documentclass{new_tlp}

\usepackage[utf8]{inputenc}
\usepackage{amsmath}
\usepackage{amssymb}
\usepackage{microtype}
\usepackage{url}
\usepackage{color}



\newcommand{\eqdef}%
{%
	\mathrel{\vbox{\offinterlineskip\ialign%
	{%
		\hfil##\hfil\cr%
		$\scriptscriptstyle\mathrm{def}$\cr%
		\noalign{\kern1pt}%
		$=$\cr%
		\noalign{\kern-0.1pt}%
	}}}%
}

\newcommand{\boldV}{\mathbf{V}}

\newcommand{\boldv}{\mathbf{v}}

\DeclareMathOperator{\rel}{\mathit{rel}}

\DeclareMathOperator{\body}{\mathit{Body}}
\DeclareMathOperator{\head}{\mathit{Head}}

\DeclareMathOperator{\mynot}{\mathit{not}}
\DeclareMathOperator{\valop}{\mathit{val}}
\newcommand{\val}[2]{\valop_{#1}(#2)}

\DeclareMathOperator{\comp}{{\rm COMP}}

\newcommand{\I}{\mathcal{I}}
\newcommand{\J}{\mathcal{J}}
\newcommand{\PP}{\mathcal{P}}
\newcommand{\SSS}{\mathcal{S}}
\newcommand{\HH}{\mathcal{H}}
\newcommand{\num}[1]{\overline{#1}}
\newcommand{\anthem}{\textsc{anthem}}
\newcommand{\anthemp}{\textsc{anthem-p2p}}

\newcommand{\cmodels}{\textsc{cmodels}}

\newcommand{\gringo}{\textsc{gringo}}
\newcommand{\vampire}{\textsc{vampire}}

\newcommand{\boldd}{\mathbf{d}}
\newcommand{\bolde}{\mathbf{e}}

\newcommand{\boldr}{\mathbf{r}}
\newcommand{\boldt}{\mathbf{t}}
\newcommand{\boldU}{\mathbf{U}}

\newcommand{\uc}{\widetilde{\forall}}

\newcommand{\Pos}[2]{\mathit{Pos}({#2},#1)}

\def\mg{mini-{\sc gringo}}
\def\modelsht
{\models_{ht}}
\def\dar{\downarrow}
\def\ar{\leftarrow}
\def\rar{\rightarrow}
\def\lrar{\leftrightarrow}
\def\no{\emph{not}}
\def\beq{\begin{equation}}
\def\eeq#1{\label{#1}\end{equation}}
\def\ba{\begin{array}}
\def\ea{\end{array}}

\newtheorem{theorem}{Theorem}
\newtheorem{proposition}{Proposition}

\newtheorem{lemma}{Lemma}

\begin{document}
\title{Locally Tight Programs}

 \author[Jorge Fandinno, Vladimir Lifschitz and Nathan Temple]{
   Jorge Fandinno\\
   University of Nebraska Omaha, USA
   \and
   Vladimir Lifschitz and Nathan Temple\\
   University of Texas at Austin, USA
 }
\maketitle

\begin{abstract}
  Program completion is a translation from the language of logic programs
  into the language of first\nobreakdash-order theories.  Its original
  definition has been extended to programs that include integer arithmetic,
  accept input, and distinguish between output predicates and auxiliary
  predicates.  For tight programs, that generalization of completion is
  known to match the stable model semantics, which is the basis of answer set
  programming.  We show that the tightness condition in
  this theorem can be replaced by a less restrictive ``local tightness''
  requirement.  From this fact we conclude that the proof assistant
  \anthemp\ can be used to verify equivalence between locally tight programs.
\end{abstract}

\section{Introduction}

Program completion \cite{cla78} is a translation from the language of
logic programs into the language of first\nobreakdash-order theories.  If, for
example, a program defines the predicate $q/1$ by the rules
\beq\ba l
q(a),\\
q(X) \ar p(X),
\ea\eeq{pr1}
then the completed definition of~$q/1$ is
\beq
\forall X(q(X)\lrar X=a \lor p(X)).
\eeq{th1}
Program~(\ref{pr1}) and formula~(\ref{th1}) express, in different
languages, the same idea: the set~$q$ consists of~$a$ and all elements of~$p$.

Fran\c{c}ois Fages \citeyear{fag94} identified a class of programs for
which the completion semantics is equivalent to the stable model
semantics.  Programs in this class are now called
\emph{tight}.  This relationship between completion and stable models plays
an important role in the theory of answer set programming (ASP).  It is
used in the design of the answer set solvers \cmodels~\cite{lie04},
{\sc assat}~\cite{lin04} and {\sc clasp}~\cite{geb07a}; and in the design of the proof
assistants \anthem\ \cite{fan20} and \anthemp\ \cite{fan23a}.

Some constructs that are common in ASP programs are not covered
by Clark's definition of program completion, and that definition had to be
generalized.  First, the original definition does not cover programs that
involve integer arithmetic.  To extend it to such programs, we
distinguish, in the corresponding first\nobreakdash-order formula, between
``general'' variables on the one hand, and variables for integers on the other
\cite[Section~5]{lif19}.  This is useful because function symbols
in a first\nobreakdash-order language are supposed to represent total functions, and
arithmetical operations are not defined on symbolic constants.

Second, in ASP programs we often find auxiliary predicate symbols, which
are not considered part of its output.
Consider, for instance, the program
\begin{align}
  \emph{in\/}(P,R,0) &\ar \emph{in\/}_0(P,R),
\label{rooms1}
\\
\emph{in\/}(P,R,T+ 1) &\ar \emph{goto\/}(P,R,T),
\label{rooms2}
\\
\{\,\emph{in\/}(P,R,T+ 1)\,\} &\ar \emph{in\/}(P,R,T)
                               \land T =  0\,..\,h- 1,
\label{rooms3}
\\
&\ar \emph{in\/}(P,R_1,T)\land \emph{in\/}(P,R_2,T)\land R_1 \neq R_2,
\label{rooms4}
\\
\emph{in\_building\/}(P,T) &\ar \emph{in\/}(P,R,T),
\label{rooms5}
\\
&\ar \no\ \emph{in\_building\/}(P,T)\land \emph{person\/}(P) \land T =  0\,..\,h,
\label{rooms6}
\end{align}%
% \beq
% \emph{in\/}(P,R,0) \ar \emph{in\/}_0(P,R),
% \eeq{rooms1}
%
% \vskip -5.5mm
% \beq
% \emph{in\/}(P,R,T+ 1) \ar \emph{goto\/}(P,R,T),
% \eeq{rooms2}
% \beq
% \{\,\emph{in\/}(P,R,T+ 1)\,\} \ar \emph{in\/}(P,R,T)
%                                \land T =  0\,..\,h- 1,
% \eeq{rooms3}
% \beq
% \ar \emph{in\/}(P,R_1,T)\land \emph{in\/}(P,R_2,T)\land R_1 \neq R_2,
% \eeq{rooms4}
% \beq
% \emph{in\_building\/}(P,T) \ar \emph{in\/}(P,R,T),
% \eeq{rooms5}
% \beq
% \ar \no\ \emph{in\_building\/}(P,T)\land \emph{person\/}(P)
%                               \land T =  0\,..\,h,
% \eeq{rooms6}
which will be used as a running example.
It describes the effect of an action---walking from one room to another---on
the location of a person.  We can
  read $\emph{in\/}(P,R,T)$ as ``person~$P$ is in room~$R$ at time~$T$,'' and
  $\emph{goto\/}(P,R,T)$ as ``person~$P$ goes to room~$R$ between times~$T$
  and~$T+1$.'' The placeholder~$h$ represents the horizon---the length of
  scenarios under consideration.  Choice rule~(\ref{rooms3})
  encodes the commonsense law of inertia for this dynamic system:
  in the absence of
  information to the contrary, the location of a person at time~$T+1$ is
  presumed to be the same as at time~$T$.  To run this program, we
  need to specify the
value of~$h$ and the input predicates \emph{person\/}/1, $\emph{in\/}_0/2$
and \emph{goto\/}/3.  The output is represented by the atoms in the stable
model that contain \emph{in\/}/3.  The predicate symbol
\emph{in\_bulding\/}/2 is auxiliary; a file containing
program~(\ref{rooms1})--(\ref{rooms6})
may contain a directive that causes the solver not to display the atoms
containing that symbol.

For this program and other programs containing auxiliary symbols, we
would like the
completion to describe the ``essential parts'' of its stable models,
with all auxiliary atoms removed.  This can be accomplished by replacing
the auxiliary symbols in Clark's completion by existentially quantified
predicate variables \cite[Section~6.2]{fan20}.  The version of
completion defined in that paper covers both integer arithmetic and
the distinction between the entire stable model and its essential part.
Standard models of the completion, in the sense of that paper,
correspond to the essential parts of the program's stable models if
the program is tight
\cite[Theorem~2, reproduced in Section~\ref{ssec:tight} below]{fan20}.

That theorem is not applicable, however, to
program~(\ref{rooms1})--(\ref{rooms6}),
because this program is not tight.  Tightness is defined as the absence
of certain cyclical dependencies between predicate symbols (see
Section~\ref{ssec:tight} for details).  The two
occurrences
of the predicate \emph{in\/}/3 in rule~(\ref{rooms3})  create a dependency
that is not allowed in a tight program.

In this article we propose the definition of a ``locally tight'' program and
show that the above-mentioned theorem by Fandinno et al.~can be extended to
programs satisfying this less restrictive condition.  Local tightness is
defined in terms of dependencies between ground atoms.  The last argument of
\emph{in\/}/3 in the head of rule~(\ref{rooms1})--(\ref{rooms6}) is $T+1$,
and the last
argument of \emph{in\/}/3 in the body is~$T$; for this reason, the
dependencies between ground atoms corresponding to this rule are not cyclic.
ASP encodings of dynamic systems \cite[Chapter~8]{lif19a}
provide many examples of this kind.
The tightness condition prohibits pretty much all uses of recursion in a
program; local tightness, on the other hand, expresses the absence of
``nonterminating'' recursion.

The result of this article shows, for example,
that the completion of program~(\ref{rooms1})--(\ref{rooms6})
correctly describes the essential parts of its stable models.

Section~\ref{sec:background} is a review of definitions and notation
related to ASP programs with arithmetic.  After defining
program completion in Section~\ref{sec:completion}, we introduce
locally tight programs and state a theorem describing the
relationship between stable models and completion under a local tightness
assumption (Section~\ref{sec:main}). In Section~\ref{sec:equiv}, that theorem
is used to justify the usability of the proof assistant \anthemp\ for
verifying equivalence between programs in some cases that are not allowed
in the original publication~\cite{fan23a}.  Proofs of the theorems stated in
these sections are based on a lemma of more general nature, which relates
stable models to completion for a class of many\nobreak-sorted
first\nobreak-order theories.  This Main Lemma is stated in
Section~\ref{sec:mainlemma} and proved in
Section~\ref{sec:prooflemma}.  Proofs of the theorems are given in
Sections~\ref{ssec:taustar.prop}--\ref{sec:proofthm3}.
In two appendices, we review terminology and
notation related to many-sorted formulas and their stable models.

\section{Background} \label{sec:background}

This review follows previous pubications on ASP programs with arithmetic
\cite{lif19,fan20,fan23a}.

\subsection{Programs} \label{ssec:programs}

The programming language mini\nobreakdash-\gringo, defined in this section, is a subset of the input language of the grounder
\gringo~\cite{geb07b}.
Most constructs included in this language are available also in the input language of the solver {\sc dlv}~\cite{leo06a}.
The description of mini\nobreakdash-\gringo\ programs below uses ``abstract   syntax,'' which disregards some details related to representing programs by strings of ASCII characters~\cite{geb15}.

\subsubsection{Syntax} \label{sssec:syntax}

We assume that three countably infinite sets of symbols are selected:
\emph{numerals}, \emph{symbolic constants}, and \emph{variables}.
We assume that a 1-1 correspondence between numerals
and integers is chosen; the numeral corresponding to an integer~$n$ is
denoted by $\num n$. 
 In examples, we take the liberty to drop the
 overline in numerals.  This convention is used, for instance, in
 rule~(\ref{rooms3}), which should be written, strictly speaking, as
$$
\{\, \emph{in\/}(P,R,T+\num 1) \,\} \ar \emph{in\/}(P,R,T) \land
T = \num 0\,..\,h-\num 1.
$$

\emph{Precomputed terms} are numerals and symbolic constants.
We assume that a total order on precomputed terms is
chosen such that for all integers~$m$ and~$n$, $\num m < \num n$ iff $m<n$.

Terms allowed in a mini-\gringo\ program are formed from
precomputed terms and variables using the absolute value
symbol $|\,|$ and six  binary operation names
$$+\quad-\quad\times\quad/\quad\backslash \quad..$$
%
% As usual in algebra, binary operators are written using infix notation,
% and the vertical bars of the absolute value symbol are written around the
% argument.
%
An \emph{atom} is a symbolic constant optionally followed by a tuple
of terms in parentheses.  A \emph{literal} is an atom
possibly preceded by one or two occurrences of \emph{not}. A \emph{comparison}
is an expression of the form $t_1\;\mathit{rel}\; t_2$,
where $t_1$, $t_2$ are terms and $\mathit{rel}$ is $=$ or one
of the comparison symbols
\beq
\neq\quad<\quad>\quad\leq\quad\geq
\eeq{comp}

A \emph{rule} is an expression of the form $\head\ar\body$,
where
 \begin{itemize}
\item
\strut$\body$ is a conjunction (possibly empty) of literals and comparisons,
and
 \item
   $\head$ is either an atom (then this is a \emph{basic rule\/}),
   or an atom in braces (then this is a
   \emph{choice rule\/}), or empty (then this is a \emph{constraint\/}).
 \end{itemize}
A rule is \emph{ground} if it does not contain variables.
A ground rule of the form $\head\ar$, where $\head$ is an
atom, will be identified with the atom $\head$ and called a \emph{fact}.

A \emph{program} is a finite set of rules.

A \emph{predicate symbol} is a pair $p/n$, where~$p$ is
a symbolic constant, and~$n$ is a nonnegative integer.  We say that~$p/n$
\emph{occurs} in a literal~$l$ if~$l$ contains an atom of the form
$p(t_1,\dots,t_n)$, and similarly for occurrences in rules and in other
syntactic expressions.

An atom $p(t_1,\dots,t_n)$ is \emph{precomputed} if the terms
$t_1,\dots,t_n$ are precomputed.

\subsubsection{Semantics} \label{sssec:semantics}

The semantics of ground terms is defined by assigning to every ground term~$t$
a finite set~$[t]$ of precomputed terms called its \emph{values}.
It is recursively defined as follows:
\begin{itemize}
\item if $t$ is a numeral or a symbolic constant, then $[t]$ is the
  singleton set~$\{t\}$;
\item if $t$ is $|t_1|$, then $[t]$ is the set of numerals $\overline{|n|}$ for all integers~$n$ such that~$\overline{n} \in [t_1]$;
 
\item if $t$ is $(t_1 + t_2)$, then $[t]$ is the set of numerals $\overline{n_1 + n_2}$ for all integers~$n_1$, $n_2$ such that~$\overline{n_1} \in [t_1]$ and $\overline{n_2} \in [t_2]$; similarly when $t$ is $(t_1 - t_2)$ or $(t_1 \times t_2)$;
\item if $t$ is $(t_1/t_2)$, then $[t]$ is the set of numerals $\overline{\mathit{round}(n_1/n_2)}$ for all integers $n_1$, $n_2$ such that $n_1 \in [t_1]$, $n_2 \in [t_2]$, and $n_2 \neq 0$;

\item if $t$ is $(t_1 \backslash t_2)$, then $[t]$ is the set of numerals $\overline{n_1 - n_2 \cdot \mathit{round}(n_1/n_2)}$ for all integers $n_1$, $n_2$ such that $n_1 \in [t_1]$, $n_2 \in [t_2]$, and $n_2 \neq 0$;
\item if $t$ is $(t_1..t_2)$, then $[t]$ is the set of numerals~$\overline{m}$ for all integers~$m$ such that,
for some integers $n_1$, $n_2$,
$n_1 \in [t_1]$, $n_2 \in [t_2]$, $n_1 \leq m \leq n_2$;
\end{itemize}
where the function $\mathit{round}$ is defined as follows:
\begin{gather*}
  \mathit{round}(n) = \begin{cases}
    \lfloor n \rfloor & \text{if $n \geq 0$}, \\
    \lceil n \rceil & \text{if $n < 0$}.
  \end{cases}
\end{gather*}
The use of this function reflects the fact that the grounder \gringo\
\cite{gringomanual} truncates non-integer quotients toward zero.
This feature of \gringo\ was not taken into account in earlier
publications~(\citeNP[Section~4.2]{geb15}; \citeNP[Section~6]{lif19}; \citeNP[Section~3]{fan20}), where the floor function was used.

For instance,
$$
[\num 7/\num 2]=\{\num 3\},\;
[\num 0\,..\,\num 2]=\{\num 0, \num 1, \num 2\},\;
[\num 2/\num 0] = [\num 2\,..\,\num 0] = \emptyset;
$$
$[\num 2+c] = [\num 2\,..\,c] = \emptyset$
if~$c$ is a symbolic constant.

For any ground terms $t_1, \dots, t_n$, by $[t_1, \dots, t_n]$ we denote
the set of tuples $r_1, \dots, r_n$ for all $r_1 \in [t_1], \dots, r_n \in
[t_n]$.

The semantics of programs is defined by rewriting rules in the syntax of
propositional logic and referring to the definition of a stable model
(answer set) of a propositional theory \cite{fer05}.  The
transformation~$\tau$ is defined as follows.
For any ground atom $p(\boldt)$, where $\boldt$ is a tuple of terms,
\begin{itemize}
\item
$\tau (p(\boldt))$ stands for $\bigvee_{\boldr \in [\boldt]} p(\boldr)$,
\item
$\tau (\mynot p(\boldt))$ stands for $\bigvee_{\boldr \in [\boldt]} \neg p(\boldr)$, and
\item
$\tau (\mynot \mynot p(\boldt))$ stands for $\bigvee_{\boldr \in [\boldt]} \neg\neg p(\boldr)$.
\end{itemize}
For any ground comparison $t_1 \rel t_2$, we define $\tau(t_1 \rel t_2)$ as
\begin{itemize}
\item $\top$ if the relation~$\rel$ holds between some~$r_1$ from $[t_1]$ and some~$r_2$ from~$[t_2]$;
\item $\bot$ otherwise.
\end{itemize}
If each of $C_1,\dots, C_k$ is a ground literal or a ground comparison, then $\tau(C_1\land \cdots \land C_k)$ stands for $\tau C_1\land \cdots \land \tau C_k$.

If~$R$ is a ground basic rule $p(\boldt) \ar \body$, then $\tau R$ is the propositional formula
\[
\textstyle
\tau(\body) \rar
\bigwedge_{\boldr\in [\boldt]} p(\boldr).
\]
If~$R$ is a ground choice rule $\{p(\boldt)\} \ar \body$, then $\tau R$ is the propositional formula
\[
\textstyle
\tau(\body) \rar \bigwedge_{\boldr \in [\boldt]}
(p(\boldr)\lor\neg p(\boldr)).
\]
If~$R$ is a ground constraint $\ar\body$, then $\tau R$ is $\neg\tau(\body)$.

An \emph{instance} of a rule is a ground rule obtained from it by substituting precomputed terms for variables.
For any program~$\Pi$, $\tau\Pi$ is the set of the formulas~$\tau R$ for all instances~$R$ of the rules of~$\Pi$.  Thus~$\tau\Pi$ is a set of
propositional combinations of precomputed atoms.

For example,~$\tau$ transforms
$\{q(X)\} \ar p(X)$
into the set of formulas 
$$p(t)\to (q(t)\lor\neg q(t))$$
for all precomputed terms~$t$.  The rule
$q(\num 0\,..\,\num 2) \ar \no\ p$
is transformed into
$$\neg p\to(q(\num 0)\land q(\num 1)\land q(\num 2)).$$

\emph{Stable models} of a program are defined as stable models of the
set of formulas obtained from it by applying the transformation~$\tau$.
Thus stable models of programs are sets of precomputed atoms.
This definition is a special case of the definition proposed
by~\citeN[Section~3]{lif19} for Abstract Gringo, except for the changes in
the treatment of the division and modulo operations mentioned above.

\subsection{Programs with input and output}\label{ssec:io}

A \emph{program with input and output}, or an~\emph{io\nobreakdash-program}, is a
quadruple
\beq
(\Pi,\emph{PH},\emph{In},\emph{Out}),
\eeq{iop}
where
\begin{itemize}
\item \emph{PH} is a finite set of symbolic constants;
\item \emph{In} is a finite set of predicate symbols,
\item \emph{Out} is a finite set of predicate symbols that is disjoint
  from \emph{In},
\end{itemize}
and $\Pi$ is a \mg\ program such that the symbols from~\emph{In} do not
occur in the heads of its rules.
Members of \emph{PH} are the \emph{placeholders} of~(\ref{iop});
members of \emph{In} are the \emph{input symbols} of~(\ref{iop});
members of \emph{Out} are the \emph{output symbols} of~(\ref{iop}).
The input symbols and output symbols of an io\nobreakdash-program are collectively
called its \emph{public symbols}.  The other predicate symbols occurring
in the rules are \emph{private}.  An \emph{input atom} is an
atom $p(t_1,\dots,t_n)$ such that~$p/n$ is an input symbol.
\emph{Output atoms}, \emph{public atoms} and \emph{private atoms}
are defined in a similar way.

For any set~$\PP$ of public atoms, $\PP^{in}$ stands for the set of input
atoms in~$\PP$.

The example discussed in the introduction corresponds to
program~(\ref{rooms1})--(\ref{rooms6}) as~$\Pi$, and
\beq
\emph{PH}=\{h\},\;
\emph{In}= \{\emph{person\/}/1, \emph{in\/}_0/2, \emph{goto\/}/3\},\;
\emph{Out}= \{\emph{in\/}/3\}.
\eeq{ex1}
This io-program will be denoted by~$\Omega_1$. The only private symbol
of~$\Omega_1$ is $\emph{in\_building\/}/2$.

A \emph{valuation} on
a set~\emph{PH} of symbolic constants is a function that maps elements
of~\emph{PH}
to precomputed terms that do not belong to~\emph{PH}.
An~\emph{input} for an io-program~(\ref{iop}) is a pair $(v,\I)$,
where
\begin{itemize}
\item $v$ is a valuation on the set \emph{PH} of its placeholdlers, and
\item $\I$ is a set of precomputed input atoms such that its members do
  not contain placeholders.
\end{itemize}
An input $(v,\I)$ represents a way to choose the values of
placeholders and the extents of input predicates: for every
placeholder~$c$, specify $v(c)$ as its value, and add the atoms~$\I$
to the rules of the program as facts.  If~$\Pi$ is a mini-\gringo\ program
then $v(\Pi)$ stands for the program obtained from~$\Pi$ by
replacing every occurrence of every constant~$c$ in the domain of~$v$
by~$v(c)$.  Using this notation, we can say that
choosing $(v,\I)$ as input for~$\Pi$ amounts to
replacing~$\Pi$ by the program $v(\Pi)\cup\I$.

About a set of precomputed atoms
we say that it is an \emph{io-model} of an
\mbox{io\nobreakdash-program}~(\ref{iop})
for an input~$(v,\I)$ if it is the set
of all public atoms of some stable model of the program
$v(\Pi)\cup\I$.  If~$\PP$ is an io-model of an io-program for an
input~$(v,\I)$ then $\PP^{in}=\I$, because
the only rules of $v(\Pi)\cup\I$ containing input symbols in the head are
the facts~$\I$.

For example, an input~$(v,\I)$ for the io-program~$\Omega_1$ can be defined
by the conditions
\beq\ba l
v(h)=\num 2,\\
\I= \{\emph{person\/}(\emph{alice\/}),\emph{person\/}(\emph{bob\/}),\\
\hskip 9mm
\emph{in\/}_0(\emph{alice\/},\emph{hall\/}),
\emph{in\/}_0(\emph{bob\/},\emph{hall\/}),\\
\hskip 9mm
\emph{goto\/}(\emph{alice\/},\emph{classroom\/},\num 0),
\emph{goto\/}(\emph{bob\/},\emph{classroom\/},\num 1)\}.
\ea\eeq{exinp}
The program $v(\Pi)\cup\I$ has a unique stable model, which consists of
the members of~$\I$, the atoms
\beq\ba l
\{\emph{in\/}(\emph{alice\/},hall,\num 0),\emph{in\/}(\emph{bob\/},hall,\num 0),\\
\hskip 2mm
\emph{in\/}(\emph{alice\/},classroom,\num 1),\emph{in\/}(\emph{bob\/},hall,\num 1),\\
\hskip 2mm
\emph{in\/}(\emph{alice\/},classroom,\num 2),\emph{in\/}(\emph{bob\/},classroom,\num 2)\},
\ea\eeq{out}
and the private atoms $\emph{in\_building\/}(p,\num i)$ for all~$p$ in
$\{\emph{alice\/},\emph{bob\/}\}$ and all~$i$ in $\{0,1,2\}$.
The io-model of~$\Omega_1$ for $(v,\I)$ consists of the members of~$\I$
and atoms~(\ref{out}).

\subsection{Two-sorted formulas and standard interpretations}
\label{ssec:formulas}

The two-sorted signature~$\sigma_0$ includes
\begin{itemize}
\item the sort \emph{general} and its subsort \emph{integer};
\item all precomputed terms of \mg\ 
 as object constants; an object constant
is assigned the sort \emph{integer} iff it is a
numeral;
\item the symbol~$|\,|$ as a unary function constant;
  its argument and value have the sort \emph{integer};
\item the symbols~$+$, $-$ and~$\times$ as binary function constants;
  their arguments and values have the sort \emph{integer};
\item predicate symbols $p/n$ as
  $n$-ary predicate constants; their arguments have the sort \emph{general};
\item symbols~(\ref{comp}) as binary predicate constants;
  their arguments have the sort \emph{general}.\footnote{The
    symbols~$/$ and~$\backslash$ are not included in this signature because
    division is not a total function on integers.
The symbol~$..$ is not included either, because intervals are not meant to be among values of variables of this first\nobreakdash-order language.}
\end{itemize}

A formula of the form $(p/n)({\bf t})$ can be written also as $p({\bf t})$.
This convention allows us to view precomputed atoms
(Section~\ref{sssec:syntax}) as sentences over~$\sigma_0$.
Conjunctions of equalities and inequalities can be abbreviated as usual
in algebra; for instance, $X=Y<Z$ stands for $X=Y\land Y<Z$.

An interpretation of the signature~$\sigma_0$ is \emph{standard} if
 \begin{itemize}
 \item[(a)] its domain of the sort \emph{general} is the set of
   precomputed terms;
 \item[(b)] its domain of the sort \emph{integer} is the set of numerals;
 \item[(c)] every object  constant represents itself;
 \item[(d)] the absolute value symbol and the binary function constants are
   interpreted as usual in arithmetic;
 \item[(e)] predicate constants~(\ref{comp}) are interpreted in accordance
   with the total order on precomputed terms chosen in the definition
   of \mg\ (Section~\ref{sssec:syntax}).
 \end{itemize}
 A standard interpretation will be uniquely determined if we specify
 the set~$\J$ of precomputed atoms to which it assigns the value
 \emph{true}.  This interpretation will be denoted by~$\J^\uparrow$.
  
 When formulas over~$\sigma_0$ are used in reasoning about io-programs with
 placeholders, we may need interpretations satisfying a condition different
 from~(c).  An  interpretation is  \emph{standard for} a set~\emph{PH}
 of symbolic constants if it satisfies
conditions~(a),~(b),~(d),~(e) above and the conditions
 \begin{itemize}
 \item[(c$'$)] every object constant that belongs to \emph{PH}
   represents a precomputed term that does not belong to~\emph{PH};
 \item[(c$''$)] every object constant that does not belong to \emph{PH}
   represents itself.
 \end{itemize}
 An interpretation that is standard for \emph{PH} will be uniquely
 determined if we specify (i) the valuation~$v$ that maps every element
 of~\emph{PH} to the precomputed term that it represents, and
 (ii)~the set~$\J$ of precomputed atoms to which it assigns the value
 \emph{true}.
 This interpretation will be denoted by~$\J^v$.

 \subsection{Representing rules by formulas}\label{ssec:taustar}

 From now on, we assume that every symbol designated as a variable in \mg\
 is among general variables of the signature~$\sigma_0$.
The transformation~$\tau^*$, described in this section, converts \mg\ rules
into formulas over~$\sigma_0$.

First we define, for every mini-{\sc gringo} term $t$,
a formula
$\val tV$ over the signature~$\sigma_0$, where $V$ is a general variable that
does not occur in~$t$.  That formula expresses, informally speaking,
that~$V$ is one of the values of~$t$.  The definition is
recursive:\footnote{\label{ft} The treatment of division here follows the paper
  \cite{fan23} that corrects a mistake found in previous publications;
  see Section~\ref{sssec:semantics}.}
\begin{itemize}
\item
if $t$ is a precomputed term or a variable then $\val tV$ is $V=t$,
\item
  if~$t$ is $|t_1|$ then $\val tV$ is
  $\exists I(\val{t_1}I\land V =|I|)$,
\item
if $t$ is $t_1\;\emph{op}\;t_2$, where \emph{op} is $+$, $-$, or $\times$
then $\val tV$ is
$$\exists I J  (\val{t_1}I \land \val{t_2}J \land V=I\;\emph{op}\;J),$$
\item
if $t$ is $t_1\,/\,t_2$ then $\val tV$ is
$$\ba l
\exists I J K (\val{t_1}I \land \val{t_2}J 
\land K\times |J|\leq |I|<(K+\num 1)\times |J|\\
\hskip 1.2cm \land\; ((I\times J \geq \num 0 \land V=K)
\lor(I\times J < \num 0 \land V=-K))),
\ea$$
\item
if $t$ is $t_1\backslash t_2$ then $\val tV$ is
$$\ba l
\exists I J K (\val{t_1}I \land \val{t_2}J
\land K\times |J|\leq |I|<(K+\num 1)\times |J|\\
\hskip 1.2cm \land\; ((I\times J \geq \num 0 \land V=I-K\times J)
\lor(I\times J < \num 0 \land V=I+K\times J))),
\ea$$
\item
if $t$ is $t_1\,..\,t_2$ then $\val tV$ is
$$\exists I J K (\val{t_1}I \land \val{t_2}J \land
    I\leq K \leq J \land V=K),$$
where $I$, $J$, $K$ are fresh integer variables.
\end{itemize}

If $\bf t$ is a tuple $t_1,\dots,t_n$ of mini-{\sc gringo} terms, and
$\bf V$ is a tuple $V_1,\dots,V_n$ of distinct general variables, then
$\val{\bf t}{\bf V}$ stands for the conjunction
$\val{t_1}{V_1} \land \cdots \land \val{t_n}{V_n}$.

The next step is to define the
transformation~$\tau^B$, which converts literals and comparisons into
formulas over~$\sigma_0$. (The superscript~$B$ reflects the fact that
this translation is close to the meaning of expressions in \emph{bodies} of
rules.)
It transforms
\begin{itemize}
\item
  $p({\bf t})$ into
  $\exists {\bf V}({\val{\bf t}{\bf V}} \land p({\bf V}))$;
\item
  $\no\ p({\bf t})$ into
  $\exists {\bf V}({\val{\bf t}{\bf V}} \land \neg p({\bf V}))$;
\item
  $\no\ \no\ p({\bf t})$  into
  $\exists {\bf V}({\val{\bf t}{\bf V}} \land \neg\neg p({\bf V}))$;
\item
$t_1 \;\mathit{rel}\; t_2$ into
$\exists V_1 V_2 (\val{t_1}{V_1} \land \val{t_2}{V_2} \land
V_1\;\mathit{rel}\; V_2)$.
\end{itemize}

If \emph{Body} is a conjunction $B_1\land B_2\land\cdots$ of literals and
comparisons then $\tau^B(\emph{Body})$ stands for the conjunction
$\tau^B(B_1)\land\tau^B(B_2)\land\cdots$.

Now we are ready to define the transformation~$\tau^*$.  It converts a
basic rule
\begin{gather}
  p({\bf t}) \ar \body
  \label{eq:basic.rule}
\end{gather}
of a program~$\Pi$ into the formula
$$\uc({\val{\bf t}{\bf V}} \land \tau^B(\body) \rar p({\bf V})),$$
where $\bf V$ is the list of alphabetically first general variables that
do not occur in~$\Pi$, and~$\uc$ denotes universal closure
(see~\ref{app:many-sorted}).  A choice rule
\beq
\{p({\bf t})\} \ar \body
\eeq{choicer}
is converted into\footnote{In some publications,
  the result of applying~$\tau^*$ to a choice rule~(\ref{choicer}) is defined
  as
$$\uc({\val{\bf t}{\bf V}} \land \tau^B(\body) \rar\neg p({\bf V})
\lor\neg p({\bf V})).$$
The difference between this formula and formula~(\ref{choicef}) is
not essential, because the two formulas are satisfied by the same
HT-interpretations (see \ref{app:ht}).}
\beq
  \uc({\val{\bf t}{\bf V}} \land \tau^B(\body) \land \neg\neg p({\bf V})
  \rar p({\bf V})),
  \eeq{choicef}
and a constraint $\ar\body$ becomes $\uc(\tau^B(\body)\to\bot)$.

By $\tau^*\Pi$ we denote the set
of sentences~$\tau^* R$ for all rules~$R$ of~$\Pi$.

The transformation~$\tau^*$ can be used in place of the transformation~$\tau$
(Section~\ref{sssec:semantics}) to describe the stable models of a
\mg\ program,
as shown by the theorem below.  Its statement refers to stable models of
many\nobreakdash-sorted theories (\ref{app:ht}) for the case
when the underlying signature is~$\sigma_0$,
and all comparison symbols are classified as extensional.
Its proof is given in~Section~\ref{ssec:taustar.prop}.

\begin{theorem}\label{thm1}
  A set~$\J$ of precomputed atoms is a stable model of a
  mini-\gringo\ program~$\Pi$ iff $\J^\uparrow$ is a stable model of
  $\tau^*\Pi$.
\end{theorem}

% Transformation~$\tau^*$ can also be used to describe io\nobreakdash-models
% as follows.

% \begin{corollary}
%   A set~$\J$ of precomputed atoms is an io\nobreakdash-model of an
%   io\nobreakdash-program~(\ref{iop}) for an input~$(v,\I)$ iff
%   there is some set~$\J'$ of private precomputed atoms such that
%   $(\J\cup \J')^v$ is a stable model of $\tau^*(\Pi\cup\I)$.
% \end{corollary}

% \begin{proof}
%   Note that~$(\J\cup \J')^v$ is a stable model of $\tau^*(\Pi\cup\I)$ 
%   iff~$(\J\cup \J')^\uparrow$ is a stable model of $v(\tau^*(\Pi\cup\I))=\tau^*(v(\Pi)\cup\I)$
%   iff~$\J\cup \J'$ is a stable model of $v(\Pi)\cup\I$ (Theorem~\ref{thm1})
%   iff~$\J$ is an io\nobreakdash-model of~(\ref{iop}) (by definition).
% \end{proof}

\section{Completion} \label{sec:completion}

This section describes the process of constructing the completion of
an io-program.  This construction involves the formulas~$\tau^*R$
for the rules~$R$ of the program, and it exploits the
special syntactic form of these formulas, as discussed below.

\subsection{Completable sets} \label{ssec:simple}

Consider a many-sorted signature~$\sigma$ (see \ref{app:many-sorted})
with its predicate constants partitioned into \emph{intensional}
symbols and \emph{extensional} symbols, so that the set of intensional
symbols is finite.
A sentence over~$\sigma$ is \emph{completable} if it has the form
\beq\uc (F\to G),\eeq{cs}
where~$G$ either contains no intensional symbols or
has the form $p({\bf V})$, where~$p$ is an intensional symbol with
argument sorts~$s_1,\dots,s_n$, and~$\bf V$ is a tuple of
pairwise distinct variables of these sorts.  A finite set~$\Gamma$ of
sentences is \emph{completable} if
\begin{itemize}
\item every sentence in~$\Gamma$ is completable, and
  \item for any two sentences $\uc (F_1\to G_1)$, $\uc (F_2\to G_2)$
    in~$\Gamma$ such that~$G_1$ and~$G_2$ contain the same intensional symbol,
    $G_1=G_2$.
\end{itemize}

It is easy to see that for any \mg\ program~$\Pi$, the set~$\tau^*\Pi$ is
completable.

The \emph{definition} of an intensional symbol~$p$ in a completable
set~$\Gamma$ is the set of members~(\ref{cs}) of~$\Gamma$
such that~$p$ occurs in~$G$.
A \emph{first-order constraint}
is a completable sentence~(\ref{cs}) such that~$G$
does not contain intensional symbols.  Thus~$\Gamma$ is the union of
the definitions of intensional symbols and a set of first-order constraints.

If the definition of~$p$ in~$\Gamma$ consists of the sentences
$\uc(F_i({\bf V})\to p({\bf V}))$ then the \emph{completed definition} of~$p$
in~$\Gamma$ is the sentence
$$\forall {\bf V} \left (p({\bf V}) \lrar
   \bigvee_i \exists {\bf U}_i\,F_i({\bf V}) \right ),$$
where~${\bf U}_i$ is the list of free variables of $F_i({\bf V})$ that
do not belong to~$\bf V$.

The \emph{completion} $\comp[\Gamma]$
of a completable set~$\Gamma$ is the conjunction of
the completed definitions of all intensional symbols of~$\sigma$ in~$\Gamma$
and the constraints of~$\Gamma$.

\subsection{Completion of a program with input and output} \label{ssec:comp}

The \emph{completion} of an io-program~(\ref{iop})
is the second-order sentence\footnote{Recall that second-order formulas
  may contain predicate variables, which can be used to form atomic formulas
  in the same way as predicate constants, 
and can be bound by quantifiers in the
same way as object variables \cite[Section~1.2.3]{lif08b}.
In many-sorted setting, each predicate variable is assigned
argument sorts, like a predicate constant (see~\ref{app:many-sorted}).
}
$\exists P_1 \cdots P_l\,C$, where~$C$ is obtained from the sentence
$\comp[\tau^*\Pi]$ by replacing all private symbols $p_1/n_1,\dots,p_l/n_l$
with distinct predicate variables $P_1,\dots,P_l$.
The completion
of~$\Omega$ will be denoted by $\comp[\Omega]$.%
\footnote{The
  definitions of $\comp[\Omega]$ in other publications
  do not refer to the translation~$\tau^*$.  But they describe
  formulas that are almost identical to the sentence $\comp[\Omega]$ defined
  here, and are equivalent to it in classical second-order logic.%
}
In building~$\comp[\Omega]$, we classify as extensional all comparison and input symbols.
All output and private symbols are intensional.

Consider, for example, the io\nobreakdash-program with the
rules
\beq
\ba l
p(a),\\
p(b),\\
q(X,Y) \ar p(X) \land p(Y),
\ea
\eeq{tpr}
without input symbols and with the output symbol $q/2$.
The translation~$\tau^*$ transforms these rules into the formulas
\begin{gather}
  \newcommand{\shorteq}{\hspace{-2pt}=\hspace{-2pt}}
  \begin{aligned}[c]
    &\forall V_1(V_1=a \to p(V_1)),\\
    &\forall V_1(V_1=b \to p(V_1)),\\
    &\forall XYV_1V_2(V_1 \shorteq  X \land V_2 \shorteq Y \land
    \exists V(V \shorteq X \land p(V)) \land\exists V(V \shorteq Y \land p(V))\to q(V_1,V_2)).%
  \end{aligned}
    \label{tpf}
\end{gather}
The completed definition of~$p/1$ in~(\ref{tpf}) is
$$\forall V_1 (p(V_1) \lrar V_1=a \lor V_1=b),$$
and the completed definition of~$q/2$  is
$$\ba l
\forall V_1V_2(q(V_1,V_2)\, \lrar\\
\qquad\quad\qquad
\exists XY(V_1=X \land V_2=Y \land \exists V(V=X \land p(V)) \land
\exists V(V=Y\land p(V)))).
\ea$$
The completion of the program is
\beq
\ba l
\!\!\!\exists P(\forall V_1(P(V_1) \lrar V_1=a \lor V_1=b)\,\land\\
\quad\;\forall V_1V_2(q(V_1,V_2)\, \lrar\\
\qquad\;\;\;
\exists XY(V_1=X \land V_2=Y \land \exists V(V=X \land P(V)) \land
\exists V(V=Y\land P(V))))).
\ea
\eeq{sof}
Formula~(\ref{sof}) can be equivalently rewritten as
$$
\exists P(\forall V_1(P(V_1) \lrar V_1=a \lor V_1=b)\land
\forall V_1V_2(q(V_1,V_2) \lrar
P(V_1) \land P(V_2))).
$$

Second-order quantifiers in completion formulas can often be eliminated.
For instance, formula~(\ref{sof}) is equivalent to the first\nobreakdash-order formula
$$
\forall V_1V_2(q(V_1,V_2) \leftrightarrow (V_1=a \lor V_1=b)\land(V_2=a \lor V_2=b)).
$$

\subsection{Review: tight programs}\label{ssec:tight}

We are interested in cases when 
the sentence $\comp[\Omega]$ can be viewed as a description of the
io-models of~$\Omega$.  This is a review of one result of this kind.

The \emph{positive predicate dependency graph} of a \mg\ program~$\Pi$
is the directed graph defined as follows.
Its vertices are the predicate symbols~$p/n$ that occur in~$\Pi$.
It has an edge from $p/n$ to $p'/n'$ iff~$p/n$ occurs in the head of a
rule~$R$ of~$\Pi$ such that $p'/n'$ occurs in an atom that is a
conjunctive term of the body of~$R$.

A \mg\ program is \emph{tight} if its positive predicate dependency graph
is acyclic.
For example, the positive predicate dependency graph of
program~(\ref{tpr}) has  a single edge, from $q/2$ to~$p/1$; this program is
tight.  The positive predicate dependency graph of
program~(\ref{rooms1})--(\ref{rooms6}) includes a self-loop at~$\emph{in}/3$;
this program is not tight.  

% For any valuation~$v$ on a set~\emph{PH} of symbolic constants and any
% set~$\J$ of precomputed atoms that do not contain symbols from \emph{PH},
% there exists a unique interpretation~$I$ of~$\sigma_0$ such that
% \begin{itemize}
% \item $I$ is standard for \emph{PH},
% \item $I$ interprets every symbolic constant~$t$ from \emph{PH} as $v(t)$, and
% \item for every precomputed atom $p({\bf t})$ that does not contain
%   symbolic constants from~\emph{PH}, $I$ satisfies $p({\bf t})$ iff
%   $p({\bf t})\in\J$.
% \end{itemize}
% This interpretation will be denoted by~$\J^v$.

Let~$\Omega$ be an io-program $(\Pi,\emph{PH},\emph{In},\emph{Out})$,
and let~$\PP$ be a set of precomputed public atoms that
do not contain placeholders of~$\Omega$.
Under the assumption that~$\Pi$ is tight,~$\PP$ is an io-model of~$\Omega$ for
an input~$(v,\I)$ iff  $\PP^v$ satisfies $\comp[\Omega]$ and $\PP^{in} = \I$
\cite[Theorem~2]{fan20}.

\section{Locally tight programs}\label{sec:main}

For any tuple $t_1,\dots,t_n$ of ground terms, $[t_1,\dots,t_n]$
stands for the set of tuples $r_1,\dots,r_n$ of precomputed terms such that
$r_1\in[t_1],\dots,r_n\in[t_n]$.

The \emph{positive dependency graph of an io\nobreakdash-program~$\Omega$ for an input
  $(v,\I)$} is the directed graph defined as follows.
Its vertices
are the ground atoms $p(r_1,\dots,r_n)$ such that $p/n$ is an output symbol
or a private symbol of~$\Omega$, and each~$r_i$ is a precomputed term
different from the placeholders of~$\Omega$.
It has an edge from $p({\bf r})$ to $p'({\bf r}')$ iff
there exists a ground instance~$R$ of one of the rules of $v(\Pi)$ such that
\begin{itemize}
\item[(a)]
  the head of~$R$ has the form $p({\bf t})$ or $\{p({\bf t})\}$,
  where ${\bf t}$  is a tuple of terms such that ${\bf r}\in[{\bf t}]$;
\item[(b)]
  the body of~$R$ has a conjunctive term of the form $p'({\bf t})$,
  where ${\bf t}$  is a tuple of terms such that ${\bf r}'\in[{\bf t}]$;
\item[(c)]
  for every conjunctive term of the body of~$R$ that
  contains an input symbol of~$\Omega$ and has the form
  $q({\bf t})$ or $\no\ \no\ q({\bf t})$, the set $[{\bf t}]$
  contains a tuple~$\bf r$ such that $q({\bf r})\in\I$;
\item[(d)]
  for every conjunctive term of the body of~$R$ that
  contains an input symbol of~$\Omega$ and has the form
  $\no\ q({\bf t})$, the set $[{\bf t}]$
  contains a tuple~$\bf r$ such that
  $q({\bf r})\not\in\I$;
\item[(e)]
  for every comparison ${t_1 \prec t_2}$ in the body of~$R$, there
  exist terms~$r_1$,~$r_2$ such that $r_1\in[t_1]$, $r_2\in[t_2]$, and
  the relation~$\prec$ holds for the pair~${(r_1,r_2)}$.
\end{itemize}

Recall that an \emph{infinite walk} in a directed graph is an infinite sequence of edges which joins a sequence of vertices.
If the positive dependency graph of~$\Omega$ for an input $(v,\I)$ has
no infinite walks, then we say that~$\Omega$ is \emph{locally tight} on
that input.

Consider, for example, the positive dependency graph of the
io-program~$\Omega_1$ (Section~\ref{ssec:io}).  Its vertices
for an input $(v,\I)$ are atoms
$\emph{in\/}(p,r,t)$ and $\emph{in\_building\/}(p,t)$, where $p$, $r$, $t$
are precomputed terms different from~$h$.  The only rules of~$\Omega_1$ that
contribute edges to the graph are~(\ref{rooms3}) and~(\ref{rooms5}), because
all predicate symbols in the bodies of~(\ref{rooms1}) and~(\ref{rooms2})
are input symbols, and rules~(\ref{rooms4}) and~(\ref{rooms6}) are constraints.
If~$R$ is~(\ref{rooms3}) then $v(R)$ is the rule
$$
\{\emph{in\/}(P,R,T+\num 1)\} \ar \emph{in\/}(P,R,T) \land
T = \num 0\,..\,v(h)\!-\!\num 1.
$$
An instance
$$
\{\emph{in\/}(p,r,t+\num 1)\} \ar \emph{in\/}(p,r,t) \land
t = \num 0\,..\,v(h)\!-\!\num 1
$$
of this rule contributes an edge if and only if~$v(h)$ is a numeral
$\num n$,~$t$ is a numeral~$\num i$, and $0\leq i<n$.
If~$R$ is~(\ref{rooms5}) then $v(R)=R$, and an instance
$$\emph{in\_building\/}(p,t) \ar \emph{in\/}(p,r,t)$$
of~$v(R)$ contributes the edge
\beq
\emph{in\_building\/}(p,t) \longrightarrow \emph{in\/}(p,r,t).
\eeq{edge}
Consequently the edges of the positive dependency graph of~$\Omega_1$
are~(\ref{edge}) and
$$
\emph{in\/}(p,r,\num{i+1}) \longrightarrow \emph{in\/}(p,r,\num i)
$$
if $v(h)=\num n$ and $0\leq i <n$.
It is clear that this graph has no infinite walks.  Consequently~$\Omega_1$
is locally tight on all inputs.

\begin{proposition}\label{prop1}
If a \mg\ program~$\Pi$ is tight then any io-program
  of the form $(\Pi,\emph{PH},\emph{In},\emph{Out})$ is locally tight for
  all inputs.
\end{proposition}

\begin{proof}
  Assume that~$\Pi$ is tight, but the positive dependency graph of
  $(\Pi,\emph{PH},\emph{In},\emph{Out})$ for some input has an infinite walk.
If that graph has an edge from $p(r_1,\dots,r_n)$
to $p'(r'_1,\dots,r'_{n'})$ then the positive predicate dependency graph
of~$\Pi$ has an edge from~$p/n$ to~$p'/n'$.  Consequently this graph
has an infinite walk as well.  But that is impossible, because this
graph is finite and acyclic.
\end{proof}

The example of program~$\Omega_1$ shows that the converse of
Proposition~\ref{prop1} does not hold.

The theorem below generalizes the property of tight programs
reviewed at the end of Section~\ref{ssec:tight}.  For any formula~$F$
over~$\sigma_0$ and any valuation~$v$,
$v(F)$ stands for the formula obtained from~$F$ by
replacing every occurrence of every constant~$c$ in the domain of~$v$
by~$v(c)$.
The proof of the following theorem is given in~Section~\ref{sec:proofs}.

\begin{theorem}\label{thm2}
\begin{samepage}
For any io\nobreakdash-program~$\Omega$, any set~$\PP$ of
precomputed public atoms that do not contain placeholders of~$\Omega$,
and any input~$(v,\I)$ for which~$\Omega$ is locally tight, the following
conditions are equivalent:
\begin{itemize}
\item[(a)] $\PP$ is an io-model of~$\Omega$ for~$(v,\I)$,
\item[(b)] $\PP^\uparrow$ satisfies $v(\comp[\Omega])$ and $\PP^{in} = \I$,
\item[(c)] $\PP^v$ satisfies $\comp[\Omega]$ and $\PP^{in} = \I$.
\end{itemize}
\end{samepage}
\end{theorem}

Take, for instance, the io-program~$\Omega_1$ and the input $(v,\I)$
defined by conditions~(\ref{exinp}).  We observed in Section~\ref{ssec:io}
that the set~$\PP$ obtained from~$\I$ by adding atoms~(\ref{out}) is an
io-model of~$\Omega_1$ for $(v,\I)$.  According to Theorem~\ref{thm2},
this observation is equivalent to the claim that the corresponding
interpretation~$\PP^\uparrow$ satisfies the sentence obtained from
$\comp[\Omega_1]$ by substituting~$\num 2$ for the placeholder~$h$.

% \begin{corollary}\label{cor1}
% For any io\nobreakdash-program~$\Omega$ that is locally tight on all inputs,
% any set~$\PP$ of precomputed public atoms that do not contain placeholders
% of~$\Omega$, and any valuation~$v$ on the set of placeholders of~$\Omega$,
% the following conditions are equivalent:  
% \begin{itemize}
%  \item[(a)] $\PP$ is an io-model of~$\Omega$ for the input of~$(v,\PP^{in})$,
%  \item[(b)] $\PP^\uparrow$ satisfies $v(\comp[\Omega])$,
% \item[(c)] $\PP^v$ satisfies $\comp[\Omega]$.
% \end{itemize}
% \end{corollary}

\section{Equivalence of io-programs} \label{sec:equiv}

We begin with an example.
Consider the io-program~$\Omega_2$, obtained from~$\Omega_1$ by replacing
inertia rule~(\ref{rooms3}) with the pair of rules
\beq
\emph{go\/}(P,T) \ar \emph{goto\/}(P,R,T),
\eeq{rooms3a}
\beq
\emph{in\/}(P,R,T+1) \ar \emph{in\/}(P,R,T) \land \no\ \emph{go\/}(P,T)
                               \land T =  0\,..\,h- 1.
                               \eeq{rooms3b}
Here $\emph{go}/2$ is a new private symbol.
These rules express commonsense inertia in a different way, using
the fact that \emph{goto} actions are the only possible
causes of change in the location of a person.  Programs~$\Omega_1$
and~$\Omega_2$ describe the same dynamic system using two different
approaches to encoding inertia.

Are these programs equivalent, in some sense?
The idea of equivalence with respect to a user guide \cite{fan23a} allows
us to make this question precise.
We present the details below.

About two io-programs we say that they are \emph{comparable} to each other
if they have the same placeholders, the same input symbols, and the same
output symbols.  For example, $\Omega_1$ and~$\Omega_2$ are comparable.
Any two comparable programs have the same inputs, because the definition of
an input for~$\Omega$ (Section~\ref{ssec:io}) refers to only two
components of~$\Omega$---to its placeholders and its input symbols.

We will define when comparable io-programs~$\Omega$,~$\Omega'$ are
equivalent to each other on a subset \emph{Dom} of their inputs.
Sets of inputs will be called \emph{domains}.  We need domains to express
the idea that two io-programs can be equivalent even if they exhibit
different behaviors on
some inputs that we consider irrelevant.  In an input for~$\Omega_1$
or~$\Omega_2$, for example, the first argument of each of the predicate
symbols $\emph{in\/}_0/2$, $\emph{goto\/}/3$ is expected to be a person;
we are not interested in inputs that contain
$\emph{in\/}_0(\emph{alice\/},\emph{hall\/})$ but do not contain
$\emph{person\/}(\emph{alice\/})$.

Domains are infinite sets, but there is a natural
way to represent some domains by
finite expressions.  An \emph{assumption} for an io-program~$\Omega$ is a
sentence over~$\sigma_0$ such that every predicate symbol~$p/n$ occurring
in it is an input symbol of~$\Omega$.  The domain \emph{defined} by an
assumption~$A$ is the set of all inputs $(v,\I)$ such that the
interpretation~$\I^v$ satisfies~$A$.  For example, the assumption
\beq
\forall P R (\emph{in\/}_0(P,R) \to \emph{person\/}(P)),
\eeq{as1}
where~$P$ and~$R$ are general variables,
defines the set~$\I$ of all inputs for~$\Omega_1$ such that
$\emph{person\/}(p)\in\I$ whenever $\emph{in\/}_0(p,r)\in\I$.

About comparable io-programs~$\Omega$,~$\Omega'$ we say that they are
\emph{equivalent} on a domain \emph{Dom} if, for every input~$(v,\I)$
from \emph{Dom},~$\Omega$ and~$\Omega'$ have the same io-models for that
input.

\anthemp\ \cite{fan23a} is a proof assistant designed for verifying
equivalence of io-programs in the sense of this definition.
It uses the following terminology.  A \emph{user guide} is a quadruple
\beq
(\emph{PH},\emph{In},\emph{Out},\emph{Dom\/}),
\eeq{ug}
where~\emph{PH}, \emph{In} and \emph{Out} are as in the definition of
an io-program (Section~\ref{ssec:io}), and \emph{Dom\/} is a domain.
The assertion that io-programs
$(\Pi,\emph{PH},\emph{In},\emph{Out})$
and
$(\Pi',\emph{PH},\emph{In},\emph{Out})$
are equivalent on a domain~\emph{Dom} can be expressed by saying that
$\Pi$ and~$\Pi'$ are equivalent with respect to user
guide~(\ref{ug}).\footnote{The definition
in the \anthemp\ paper looks different from the one
given here, but the two definitions are equivalent.  Instead of the condition
``$\Omega$ and~$\Omega'$ have the same io-models''
Fandinno et al.~require that~$\Pi$
  and~$\Pi'$ have the same external behavior.  The external behavior
  of a \mg\ program~$\Pi$ for a user guide~(\ref{ug}) and an input $(v,\I)$
  can be described as the collection of all sets that can be
  obtained from io-models of $(\Pi,\emph{PH},\emph{In},\emph{Out})$ by
  removing the input atoms~$\I$.}

We can say, for example, that replacing rule~(\ref{rooms3})
in \mg\ program (\ref{rooms1})--(\ref{rooms6}) with rules~(\ref{rooms3a})
is an equivalent transformation for the user guide~(\ref{ug}) with~\emph{PH},
\emph{In} and \emph{Out} defined by formulas~(\ref{ex1}), and with
\emph{Dom} defined by the conjunction
of assumption~(\ref{as1}) and the assumption
\beq
\forall P R\, T(\emph{goto\/}(P,R,T) \to \emph{person\/}(P)).
\eeq{as2}

Theorem~\ref{thm3} below shows that using the
\anthemp\ algorithm for verifying equivalence of io-programs
can be justified under a local tightness condition.
If an io-program is locally tight for all inputs satisfying an assumption~$A$
then we say that it is locally tight \emph{under the assumption}~$A$.
If two comparable io-programs are equivalent on the domain defined by an
assumption~$A$ then we say that they are equivalent \emph{under the
  assumption}~$A$.
The proof of the following theorem is given in~Section~\ref{sec:proofthm3}.

\begin{theorem}\label{thm3}
  Let~$\Omega$,~$\Omega'$ be comparable io-programs with
  placeholders~\emph{PH} that are locally tight under an assumption~$A$.
  Programs~$\Omega$,~$\Omega'$ are equivalent to each other under the
  assumption~$A$ iff the sentence
\beq
A\to(\comp[\Omega] \lrar\comp[\Omega'])
\eeq{thmf}
is satisfied by all interpretations that are standard for~\emph{PH}.
\end{theorem}

From Theorem~\ref{thm3} we can conclude
that equivalence between~$\Omega$ and~$\Omega'$ will be established if
formula~(\ref{thmf}) is derived from a set of axioms that are satisfied by all
interpretations that are standard for~\emph{PH}. This is how
\anthemp\ operates \cite[Sections~7]{fan23a}.

The claim that~$\Omega_1$ is equivalent to $\Omega_2$ assuming~(\ref{as1})
and~(\ref{as2}) has been verified by
\anthemp, with the resolution theorem prover \vampire\ \cite{vor13} used
as the proof engine.
We saw in Section~\ref{sec:main} that $\Omega_1$ is locally tight for all
inputs.  The positive dependency graph of~$\Omega_2$ has
additional edges, from $\emph{go\/}(p,t)$ to $\emph{goto\/}(p,r,t)$.  It is
clear that this graph has no infinite walks either,
so that~$\Omega_2$ is locally tight for all inputs as well.

The process of verifying equivalence in this case was interactive---we helped
\vampire\ find a proof by suggesing useful lemmas and an induction axiom,
as usual in such experiments~(\citeNP[Section~6]{fan23a}; \citeNP[Figure~5]{fan20}).
Induction was needed to prove the lemma
$$\forall P R\, T(\emph{in\/}(P,R,T) \to \emph{person\/}(P)),$$
with~(\ref{as1}) used to justify the base case.

\section{Main lemma} \label{sec:mainlemma}

The Main Lemma, stated below, relates the completion of a completable
set of sentences (Section~\ref{ssec:simple}) to its stable models
(\ref{app:ht}).

As in Section~\ref{ssec:simple}, $\sigma$ stands here for a many-sorted
signature with its predicate constants partitioned into intensional
and extensional.
 The symbols~$\sigma^I$, $\boldd^*$, used in this section, are
 introduced in \ref{app:many-sorted}; $I^\dar$ is defined in \ref{app:ht}.

We define, for an interpretation~$I$ of~$\sigma$ and a
sentence~$F$ over~$\sigma^I$, the set $\Pos{I}F$ of \emph{(strictly)
  positive atoms of~$F$
  with respect to~$I$}.  Elements of this set are formulas
of the form $p(\boldd^*)$.
This set is defined recursively, as follows.  If~$F$
does not contain intensional symbols or is not satisfied by~$I$
then $\Pos{I}F = \emptyset$. Otherwise,
\begin{itemize}
\item[(i)] $\Pos{I}{p(\boldt)} = \{p((\boldt^I)^*)\}$;
\item[(ii)] $\Pos{I}{F_1 \wedge F_2} = \Pos{I}{F_1 \vee F_2} = \Pos{I}{F_1} \cup
  \Pos{I}{F_2}$;
    \item[(iii)] $\Pos{I}{F_1 \to F_2} = \Pos{I}{F_2}$;
    \item[(iv)] $\Pos{I}{\forall X F(X)} = \Pos{I}{\exists X F(X)} =
      \bigcup_{d\in|I|^s} \Pos{I}{F(d^*)}$ if~$X$ is a variable of sort~$s$.
\end{itemize}

It is easy to check by induction on~$F$ that $\Pos{I}F$ is a subset
of the set~$I^\dar$.

An \emph{instance} of a completable sentence $\uc(F \to G)$ over~$\sigma^I$
is a sentence obtained from $F\to G$ by substituting names~$d^*$ for its free
variables.
For any completable set~$\Gamma$ of sentences over~$\sigma^I$,
the \emph{positive dependency graph}
$G^{sp}_{I}(\Gamma)$ is the directed graph defined as follows.
Its vertices are elements of~$I^\dar$.
It has an edge from~$A$ to~$B$ iff, for some instance
$F\to G$ of a member of~$\Gamma$,
$A\in\Pos{I}{G}$ and $B\in\Pos{I}{F}$.

\medskip\noindent\emph{Main Lemma}

\nobreak\noindent
For any interpretation~$I$ of~$\sigma$
  and any completable set~$\Gamma$ of sentences over~$\sigma^I$
  such that the graph~$G^{sp}_I(\Gamma)$ has no infinite walks,~$I$
  is a stable model of~$\Gamma$ iff~$I\models\comp[\Gamma]$.
  \medskip
  
  Consider, for example, the signature consisting of a single sort,
  the object constants~$a$,~$b$, and the intensional unary predicate
  constant~$p$. Let~$\Gamma$ be
  $$\{\forall V(V=a\to p(V)),\ \forall V(V=b \land p(V) \to p(V))\},$$
  and let the interpretations~$I$, $J$ be defined by the conditions
  $$\ba c
  |I|=|J|=\{0,1,\dots\};\ a^I=a^J=0;\ b^I=b^J=1;\\
p^I(n)=\emph{true}\hbox{ iff }n=0;\ p^J(n)=\emph{true}\hbox{ iff }n\in\{0,1\}.
\ea$$
Instances of the members of~$\Gamma$ are implications of the forms
$$n^*=a\to p(n^*)\quad\hbox{and}\quad n^*=b \land p(n^*) \to p(n^*)$$
($n=0,1,\dots$).
The set~$I^\dar$ of vertices of the graph $G^{sp}_I(\Gamma)$ is~$\{p(0^*)\}$.
This graph has no edges, because, for every~$n$,
$$\Pos{I}{n^*=a}=\Pos{I}{n^*=b \land p(n^*)}=\emptyset.$$
The set~$J^\dar$ of vertices of $G^{sp}_J(\Gamma)$ is
$\{p(0^*),p(1^*)\}$.  This graph has one edge---the self-loop at $p(1^*)$,
because
%%% CHECK %%%
$$\Pos{J}{1^*=b \land p(1^*)}=\Pos{J}{p(1^*)} =\{p(1^*)\}.$$
Consequently the Main Lemma is applicable to~$I$, but not to~$J$.
It asserts that~$I$ is a stable model of~$\Gamma$ (which is true) iff~$I$
satisfies the completion
$$\forall V(p(V)\lrar V=a\lor(V=b\land p(b))$$
of~$\Gamma$ (which is true as well).  The interpretation~$J$, on the
other hand, satisfies the completion of~$\Gamma$, but it
is not a stable model.

\medskip
The Main Lemma is similar to two results published earlier
(\citeNP[Theorem~11]{fer09}; \citeNP[Corollary~4]{lee11b}).
The former refers to dependencies between predicate constants,
rather than ground atoms.  The latter is applicable to formulas of
different syntactic form.

About an interpretation~$I$ of~$\sigma$ we say that it is
\emph{semi-Herbrand} if
\begin{itemize}
\item all elements of its domains are object
  constants of~$\sigma$, and
  \item $d^I=d$ for every such
element~$d$.\footnote{This is weaker than the condition defining
  Herbrand interpretations; some ground terms over~$\sigma$ may be
  different from elements of the domains~$|I|^s$.}
\end{itemize}
For example, every interpretation of~$\sigma_0$ that is standard, or
standard for some some set~$PH$ (Section~\ref{ssec:formulas}),
is semi\nobreakdash-Herbrand.

If~$I$ is semi-Herbrand then the graph $G^{sp}_I(\Gamma)$
 can be modified by replacing each vertex $p({\bf d}^*)$
by the atom $p({\bf d})$ over the signature~$\sigma$.  This modified
positive dependency graph is isomorphic to $G^{sp}_I(\Gamma)$ as
defined above.  It can be
defined independently, by replacing clauses~(i) and~(iv) in the definition
of $\Pos I F$ above by
\begin{itemize}
\item[(i)$'$] $\Pos{I}{p(\boldt)} = \{p(\boldt^I)\}$;
\item[(iv)$'$] $\Pos{I}{\forall X F(X)} = \Pos{I}{\exists X F(X)} =
      \bigcup_{d\in|I|^s} \Pos{I}{F(d)}$ if~$X$ is a variable of sort~$s$.
\end{itemize}

\section{Proof of the Main Lemma}\label{sec:prooflemma}

The Main Lemma expresses a property of the completion operator, and its
proof below consists of two parts.  We first prove a similar
property of pointwise stable models, defined in~\ref{app:ht}
(Lemma~\ref{l3}); then we relate pointwise stable models to completion.

As in Section~\ref{sec:mainlemma}, $\sigma$ stands for a many-sorted
signature with its predicate constants subdivided into intensional and
extensional.

\begin{lemma}\label{l1}
  For any HT\nobreakdash-interpretation $\langle\HH,I \rangle$ and any
  sentence~$F$ over~$\sigma^I$, if $I \models F$ and 
$\Pos{I}F\subseteq \HH$
  then $\langle\HH,I \rangle\modelsht F.$
\end{lemma}

\begin{proof}
By induction on the size of~$F$.
\emph{Case~1}: $F$ does not contain intensional symbols.  The
assertion of the lemma follows from Proposition~\ref{properties}(b)
in~\ref{app:ht}.  \emph{Case~2}:~$F$
contains an intensional symbol. Since ${I\models F}$,
the set $\Pos{I}F$ is determined by the recursive clauses in
the definition of~$\mathit{Pos}$.  \emph{Case~2.1:}~$F$ is $p(\boldt)$,
where $p$ is intensional.  Then, the assumption
$\Pos{I}F\subseteq \HH$ and the
claim $\langle\HH,I\rangle\modelsht F$
turn into the condition $p((\boldt^I)^*)\in \HH$.
\emph{Case~2.2:}~$F$ is $F_1\land F_2$.  Then, from the assumption
${I\models F}$ we conclude that ${I\models F_i}$
for $i=1,2$.  On the other hand,
$$\Pos{I}{F_i}\subseteq\Pos{I}F\subseteq \HH.$$
By the induction hypothesis, it
follows that \hbox{$\langle\HH,I\rangle\modelsht F_i$},
and consequently ${\langle\HH,I\rangle\modelsht F}$.
\emph{Case~2.3:}~$F$ is \hbox{$F_1\lor F_2$}.  Similar to Case~2.2.
\emph{Case~2.4:}~$F$ is $F_1\to F_2$.  Since ${I\models F_1\to F_2}$,
we only need to check that
% \begin{center}
${\langle\HH,I\rangle \not\modelsht F_1}$
or
${\langle\HH,I\rangle\modelsht F_2}$.
% \end{center}
\emph{Case~2.4.1:} $I\models F_1$.
Since $I\models F_1\to F_2$, it follows that
\hbox{$I\models F_2$}.
On the other hand,
$$\Pos{I}{F_2}=\Pos{I}F\subseteq \HH.$$
By the induction hypothesis, it follows that
$\langle\HH,I\rangle\modelsht F_2$.
\emph{Case~2.4.2:} $I\not\models F_1$.
By Proposition~\ref{properties}(a) in~\ref{app:ht}, it follows that
$\langle\HH,I\rangle\not\modelsht F_1$.
\emph{Case~2.5:} $F$ is $\forall X\,G(X)$, where~$X$ is a variable of
sort~$s$.  Then, for every element~$d$ of~$|I|^s$, $I\models G(d^*)$
and
$$\Pos{I}{G(d^*)}\subseteq \Pos{I}F\subseteq \HH.$$
By the induction hypothesis, it follows that
\hbox{$\langle\HH,I\rangle\modelsht G(d^*)$}.  We can conclude that
\hbox{$\langle\HH,I\rangle\modelsht \forall X\,G(X)$}.
\emph{Case~2.6:} $F$ is $\exists X\,G(X)$.  Similar to Case~2.5.
\end{proof}

The definition of $G^{sp}_{I}(\Gamma)$ in Section~\ref{sec:mainlemma}
is restricted to the case when~$\Gamma$ is a completable set of sentences over~$\sigma^I$.
It can be generalized to arbitrary sets of sentences as follows.
A \emph{rule subformula} of a formula~$F$ is an
occurrence of an implication in~$F$ that does not belong to the
antecedent of any implication
(\citeNP[Section~7.3]{fer09}; \citeNP[Section~3.3]{lee11b}).
Let~$I$ be an interpretation of an arbitrary signature~$\sigma$, and let~$\Gamma$ be a set of
sentences over~$\sigma^I$.  All vertices of $G^{sp}_{I}(\Gamma)$ are
elements of~$I^\dar$.
The graph has an edge from~$A$ to~$B$ iff, for some sentence
$F_1\to F_2$ obtained from a rule subformula of a member of~$\Gamma$
by substituting names~$d^*$ for its
free variables, $A\in\Pos{I}{F_2}$ and $B\in\Pos{I}{F_1}$.

For any sentence~$F$, $G^{sp}_{I}(F)$ stands for  $G^{sp}_{I}(\{F\})$.

\begin{lemma}\label{l2}
  For any HT\nobreakdash-interpretation $\langle\HH,I \rangle$, any
  atom $M$ in $I^\dar \setminus \HH$, and any sentence~$F$
  over~$\sigma^I$, if
\begin{itemize}
\item[(i)]
  for every edge~$(M,B)$ of~$G^{sp}_{I}(F)$, $B\in \HH$,
\item[(ii)]
  $M\in\Pos{I}{F}$, and
\item[(iii)]
  $\langle\HH,I \rangle\modelsht F$,
\end{itemize}
then $\langle I^\dar \setminus\{M\},I\rangle\modelsht F$.
\end{lemma}

\begin{proof}
By induction on the size of~$F$.
Sentence~$F$ is neither atomic nor~$\bot$.  Indeed, in that case~$F$
would be an atomic sentence of the form $p(\boldt)$, where~$p$ is
intensional, because, by~(ii), $\Pos{I}F$ is non-empty.
Then, from~(iii), $p((\boldt^I)^*)\in \HH$.  On the
other hand, $\Pos{I}F$ is $\{p((\boldt^I)^*)\}$,
and from~(ii) we conclude that $M=p((\boldt^I)^*)$.  This contradicts the
assumption that $M\in I^\dar \setminus \HH $.  Thus five cases are
possible.

\emph{Case~1:} $F$ is $F_1\land F_2$.  From~(iii) we can conclude that
$\langle\HH,I \rangle\modelsht F_i$ for $i=1,2$.
It is sufficient to show that
$\langle I^\dar \setminus\{M\}, I\rangle\modelsht F_i$;
then the conclusion that
\hbox{$\langle I^\dar \setminus\{M\}, I\rangle\modelsht F$}
will follow.
\emph{Case~1.1:} $M\in\Pos{I}{F_i}$, so that
formula~$F_i$ satisfies condition~(ii). That formula satisfies
condition~(i) as well, because
$G^{sp}_{I}(F_i)$ is a subgraph of $G^{sp}_{I}(F)$,
and it satisfies condition~(iii).  So the conclusion
$\langle I^\dar \setminus\{M\}, I\rangle\modelsht F_i$ follows by
the induction hypothesis.  \emph{Case~1.2:} $M\not\in\Pos{I}{F_i}$.
Since $\Pos{I}{F_i}$ is a subset of~$I^\dar$,
\hbox{$\Pos{I}{F_i}\subseteq I^\dar \setminus\{M\}$}.
On the other hand, from the fact that
$\langle\HH,I \rangle\modelsht F_i$ we conclude, by
Proposition~\ref{properties}(a), that $I\models F_i$.  By Lemma~\ref{l1}, it
follows that
\hbox{$\langle I^\dar\setminus\{M\},I\rangle\modelsht F_i$}.

\emph{Case~2:} $F$ is $F_1\lor F_2$.  From~(iii) we can conclude that
$\langle\HH,I \rangle\modelsht F_i$ for $i=1$ or $i=2$.
It is sufficient to show that
$\langle I^\dar\setminus\{M\},I\rangle\modelsht F_i$;
then the conclusion that
$\langle I^\dar\setminus\{M\},I\rangle\modelsht F$ will follow.
The reasoning is the same as in Case~1.

\emph{Case~3:} $F$ is $F_1\to F_2$.  Then,
$ %\beq
\Pos{I}F=\Pos{I}{F_2}
$. %\eeq{e1}
By~(ii), it follows that
\beq
M\in\Pos{I}{F_2}.
\eeq{e1a}
On the other hand,
$F$ is a rule subformula of itself, so that
for every atom~$B$ in $\Pos{I}{F_1}$, $(M,B)$ is an edge of
the graph $G^{sp}_{I}(F)$.  By~(i), it follows that every such
atom~$B$ belongs to~$\HH$.  Consequently
\beq
\Pos{I}{F_1}\subseteq \HH.
\eeq{e2}
\emph{Case~3.1:} $\langle\HH,I \rangle\modelsht F_2$, so
that~$F_2$ satisfies condition~(iii).  Since $G^{sp}_{I}(F_2)$ is
a subgraph of $G^{sp}_{I}(F)$, $F_2$ satisfies condition~(i) as
well.  By~(\ref{e1a}), $F_2$ satisfies
condition~(ii).  Then, by the induction hypothesis,
$\langle I^\dar\setminus\{M\},I\rangle\modelsht F_2$.
Consequently $\langle I^\dar\setminus\{M\},I\rangle\modelsht F$.
\emph{Case~3.2:} $\langle\HH,I \rangle\not\modelsht F_2$.
Then, in view of~(iii),
$\langle\HH,I \rangle\not\modelsht F_1$.
From this fact and formula~(\ref{e2}) we can conclude, by Lemma~\ref{l1}, that
$I \not\models F_1$.  By Proposition~\ref{properties}(a), it follows that
\hbox{$\langle I^\dar\setminus\{M\},I\rangle\not\modelsht F_1$},
which implies
$\langle I^\dar\setminus\{M\},I\rangle\modelsht F$.

\emph{Case~4:} $F$ is $\forall X\,G(X)$.  From~(iii) we can conclude that
for every $d$ in the domain $|I|^s$, where~$s$ is the sort of~$X$,
$\langle\HH,I \rangle\modelsht G(d^*)$.
It is sufficient to show that
$\langle I^\dar\setminus\{M\},I\rangle\modelsht G(d^*)$;
then the conclusion that
$\langle I^\dar\setminus\{M\},I\rangle\modelsht F$ will follow.
\emph{Case~4.1:} $M\in\Pos{I}{G(d^*)}$, so that
formula~$G(d^*)$ satisfies condition~(ii). That formula satisfies
condition~(i) as well, because
$G^{sp}_{I}(G(d^*))$ is a subgraph of $G^{sp}_{I}(F)$,
and it satisfies condition~(iii).  So the conclusion
${\langle I^\dar\setminus\{M\},I\rangle\modelsht G(d^*)}$ follows by
the induction hypothesis.
\emph{Case~4.2:} $M\not\in\Pos{I}{G(d^*)}$.
Since $\Pos{I}{G(d^*)}$ is a subset of~$I^\dar$,
we can conclude that
$\Pos{I}{G(d^*)}\subseteq I^\dar \setminus\{M\}$.
On the other hand, from the fact that
$\langle\HH,I \rangle\modelsht G(d^*)$ we conclude, by
Proposition~\ref{properties}(a), that $I\models G(d^*)$.  By Lemma~\ref{l1}, it
follows that
$\langle I^\dar\setminus\{M\},I\rangle\modelsht G(d^*)$.

\emph{Case~5:} $F$ is $\exists X\,G(X)$.  From~(iii) we can conclude that
for some $d$ in the domain $|I|^s$, where~$s$ is the sort of~$X$,
$\langle\HH ,I\rangle\modelsht G(d^*)$.
It is sufficient to show that
\hbox{$\langle I^\dar\setminus\{M\},I\rangle\modelsht G(d^*)$};
then the conclusion that
\hbox{$\langle I^\dar\setminus\{M\},I\rangle\modelsht F$}
will follow.  The reasoning is the same as in Case~4.
\end{proof}

\begin{lemma}\label{l3}
  For any interpretation~$I$ of~$\sigma$
  and any set~$\Gamma$ of sentences over~$\sigma^I$ such that the
  graph~$G^{sp}_I(\Gamma)$ has no infinite walks, $I$ is a stable
  model of~$\Gamma$ iff~$I$ is pointwise stable.
\end{lemma}

\begin{proof}
We need to show that if $I$ is a model of~$\Gamma$
such that the graph~$G^{sp}_{I}(\Gamma)$ has no infinite walks, and
there exists a proper subset~$\HH$ of~$I^\dar$ such that $\langle\HH,I\rangle$ 
satisfies~$\Gamma$, then a subset with this property can be
obtained from~$I^\dar$ by removing a single element.

The set $I^\dar \setminus \HH$ contains an atom~$M$ such that
for every edge~$(M,B)$ of the graph $G^{sp}_{I}(\Gamma)$,
$B\not\in I^\dar \setminus \HH$.  Indeed, otherwise
this graph would have an infinite walk consisting of
elements of $I^\dar \setminus \HH$. 
On the other hand, for every such edge,
$B\in I^\dar$.  Indeed, from the definition of the
graph~$G^{sp}_I(\Gamma)$ we see that for every edge~$(M,B)$ of that graph,
$B$ belongs to the set $\Pos{I}F$ for some sentence~$F$, and that set is
contained in~$I^\dar$.  Consequently for every edge~$(M,B)$ of
$G^{sp}_{I}(\Gamma)$, $B\in \HH$.

We will show that $\langle I^\dar\setminus\{M\},I\rangle$
satisfies~$\Gamma$.  Take a sentence~$F$ from~$\Gamma$.
\emph{Case~1:} \hbox{$M\in\Pos{I}F$}.
Then, condition~(ii) of Lemma~\ref{l2} is
satisfied for the HT\nobreakdash-interpretation $\langle\HH,I\rangle$.  Condition~(i)
is satisfied for this HT\nobreakdash-interpretation as well, because
$G^{sp}_{I}(F)$ is a subgraph of $G^{sp}_{I}(\Gamma)$; furthermore,
condition~(iii) is satisfied because $\langle\HH,I\rangle$ satisfies~$\Gamma$.
Consequently
$\langle I^\dar\setminus\{M\},I\rangle\modelsht F$ by Lemma~\ref{l2}.
\emph{Case~2:} $M\not\in\Pos{I}F$.   Then,
$\Pos{I}F\subseteq  I^\dar\setminus\{M\}$.  Since
$I \models F$, we can conclude that
$\langle I^\dar\setminus\{M\},I\rangle\modelsht F$
by Lemma~\ref{l1}.
\end{proof}

A model~$I$ of a set~$\Gamma$ of completable sentences over~$\sigma$ is
\emph{supported} if for every atom $p(\boldd^*)$ in~$I^\dar$ there exists an
instance $F\to p(\boldd^*)$ of a member of~$\Gamma$
such that~$I \models F$.

\begin{lemma}\label{lb}
  Every pointwise stable model of a completable set of sentences is supported.
\end{lemma}

\begin{proof}
Let~$I$ be a pointwise stable model of a completable set~$\Gamma$ of sentences.
Take an atom $p(\boldd^*)$ from~$I^\dar$.  We need to find an instance
$F\to p(\boldd^*)$ of a completable sentence from~$\Gamma$ such that~$I\models F$.

By the definition of a pointwise stable model,
$\langle I^\dar\setminus\{p(\boldd^*)\},I\rangle$ does not satisfy~$\Gamma$.
Then, one of the completable sentences from~$\Gamma$ has an instance $F\to G$
such that
\beq
\hbox{$\langle I^\dar\setminus\{p(\boldd^*)\},I\rangle\not\modelsht F\to G$}.
\eeq{C9}
We will show that this instance has the required properties.
Since $I$ is a model of~$\Gamma$,
\beq
I\models F\to G.
\eeq{C0}
From~(\ref{C9}) and~(\ref{C0}) we conclude that
\beq
\hbox{$\langle I^\dar\setminus\{p(\boldd^*)\},I\rangle\modelsht F$}
\eeq{C3}
and
\beq
\hbox{$\langle I^\dar\setminus\{p(\boldd^*)\},I\rangle\not\modelsht G$}.
\eeq{C2}
From (\ref{C3}) and Proposition~\ref{properties}(a), $I\models F$.  Then, in view of (\ref{C0}),
\beq
% \hbox{$\langle J, S,S\rangle\modelsht G$}.
I\models G
\eeq{C4}
From~(\ref{C2}),~(\ref{C4}) and Proposition~\ref{properties}(c)
we can conclude that formula~$G$d
contains~$p$, so that it has the form $p(\bolde^*)$ for some tuple of domain elements~$\bolde$.
Then, from (\ref{C4}),
$p(\bolde^*)\in I^\dar$, and
from (\ref{C2}), $p(\bolde^*)\not\in I^\dar \setminus\{p(\boldd^*)\}$.
Consequently $p(\bolde^*)$ is $p(\boldd^*)$, so that $\bolde=\boldd$.
\end{proof}

\begin{lemma}\label{l4}
  For any interpretation~$I$ of~$\sigma$
  and any completable set~$\Gamma$ of sentences over~$\sigma$ such that the
  graph~$G^{sp}_I(\Gamma)$ has no infinite walks, $I$ is a supported
  model of~$\Gamma$ iff~$I$ is stable.
\end{lemma}

\begin{proof}
The if part follows from Lemmas~\ref{l3} and~\ref{lb}.  For the only if part, consider
a supported model $I$ of a completable set~$\Gamma$ of sentences
such that the graph $G^{sp}_{I}(\Gamma)$ has no infinite walks; we need to
prove that $I$ is stable.  According to Lemma~\ref{l3}, it is sufficient to
check that~$I$ is pointwise stable.

Take any atom~$M$ in~$I^\dar$; we will show that
$\langle I^\dar\setminus\{M\}, I\rangle$ is not an HT\nobreakdash-model of~$\Gamma$.
Since $I$ is supported, one of the completable sentences in~$\Gamma$ has an
instance ${F\to M}$ such that ${I\models F}$.  Atom~$M$ does not belong
to $\Pos{I}F$, because otherwise $M,M,\dots$ would be an
infinite walk in $G^{sp}_{I}(\Gamma)$.  Since the set $\Pos{I}F$
is a subset of~$I^\dar$, we can conclude that it is a subset of $I^\dar\setminus\{M\}$.
By Lemma~\ref{l1}, it follows that
$\langle I^\dar\setminus\{M\}, I\rangle\modelsht F$.
Therefore $\langle I^\dar\setminus\{M\}, I\rangle\not\modelsht F\to M$.
\end{proof}

\begin{lemma}\label{lc}
  For any interpretation~$I$ of~$\sigma$
  and any completable set~$\Gamma$ sentences over~$\sigma$,
  $I$ is a supported model of~$\Gamma$ iff $I\models\comp[\Gamma]$.
\end{lemma}

\begin{proof}
Let the sentences defining an intensional predicate symbol~$p$ in~$\Gamma$ be
\beq
\forall \boldU_i\boldV(F_i(\boldU_i,\boldV)\to p(\boldV))\qquad ( i=1,\dots,k).
\eeq{r1}
The completed definition of~$p$ is
$$%\beq
\forall {\boldV}\left(p({\boldV}) \lrar \bigvee_{i = 1}^k
  \exists \boldU_i \, F_i(\boldU_i,\boldV)
\right).
$$%\eeq{cd1}
Hence, an interpretation $I$ of~$\sigma$ satisfies $\comp[\Gamma]$ iff the following three conditions hold:
\begin{itemize}
\item[(a)]
  for every intensional~$p$, $I$ satisfies the sentence
$$%\beq
\forall {\boldV}\left(p({\boldV}) \to \bigvee_{i = 1}^k
  \exists \boldU_i \, F_i(\boldU_i,\boldV)
\right);
$$%\eeq{cd2}
\item[(b)]
for every intensional~$p$ and for every~$i$, $I$ satisfies the sentence
\beq
\forall {\boldV}(
  \exists \boldU_i \, F_i(\boldU_i,\boldV)
\to p({\boldV}));
\eeq{cd3}
\item[(c)]
$I$ satisfies all constraints of~$\Gamma$.
\end{itemize}
Since~(\ref{cd3}) is equivalent to~(\ref{r1}),
$I$ is a model
of~$\Gamma$ if and only if conditions~(b) and~(c) hold.  It remains to
check that condition~(a) holds if and only if the model $I$ is supported.

Condition~(a) can be expressed by saying that
\medskip

\begin{center}
for every atom $p(\boldd^*)$ in~$I^\dar$ there exists~$i$ such that
$
I\models
\exists \boldU_i \, F_i(\boldU_i,\boldd^*),
$
\end{center}

\medskip\noindent
or, equivalently, that

\medskip\begin{center}
  for every atom $p(\boldd^*)$ in~$I^\dar$ there exist~$i$ and a
  tuple $\boldd_i$ of domain elements\\
  such that
  $I\models F_i(\boldd_i^*,\boldd^*)$.
\end{center}

\medskip\noindent
Since the members of~$\Gamma$ defining~$p$ are sentences of form~(\ref{r1}),
the last condition expresses that
$I$ is a supported model of~$\Gamma$.
\end{proof}

The assertion of the Main Lemma follows from Lemmas~\ref{l4} and~\ref{lc}.

\section{Proof of Theorem~\ref{thm1}}\label{ssec:taustar.prop}

The proof of Theorem~\ref{thm1} refers to infinitary
propositional formulas and the strong equivalence relation between them
\cite{har17}.

Translations~$\tau$ and~$\tau^*$ are related by a
third translation~$F\mapsto F^{\rm prop}$
\cite[Section~5]{lif19}, which transforms sentences over~$\sigma_0$ into
infinitary propositional combinations of precomputed atoms.
This translation is defined as follows:
\begin{itemize}
  \item if $F$ is $p(t_1,\dotsc,t_n)$, then $F^{\rm prop}$ is obtained from $F$ by replacing each~$t_i$ by the value obtained after evaluating all arithmetic functions in~$t_i$;
  
  \item if $F$ is $t_1 \,\mathit{rel}\, t_2$, then $F^{\rm prop}$ is $\top$ if the values of $t_1$ and $t_2$ are in the relation $\mathit{rel}$,  and $\bot$ otherwise;

  \item $\bot^{\rm prop}$ is $\bot$;

  \item $(F \odot G)^{\rm prop}$ is $F^{\rm prop} \odot G^{\rm prop}$ for every binary connective~$\odot$;
  
  \item $(\forall X\, F(X))^{\rm prop}$ is the conjunction of the formulas $F(r)^{\rm prop}$
    over all precomputed terms $r$ if $X$ is a general variable, and
    over all numerals $r$ if $X$ is an integer variable;

  \item $(\exists X\, F(X))^{\rm prop}$ is the disjunction of the formulas $F(r)^{\rm prop}$
    over all precomputed terms $r$ if $X$ is a general variable, and over
    all numerals $r$ if $X$ is an integer variable.
\end{itemize}

This translation is similar to the grounding operation defined
by~\citeN[Section~2]{tru12}.
The following proposition, analogous to Proposition~2 from Truszczynski's
paper and to Proposition~3 by~\citeN{lif19}, relates the meaning of a sentence to the meaning of its propositional translation. 
It differs from the last result in view of the fact that the division and
modulo operations are treated here in a different way (see
Section~\ref{sssec:semantics} and Footnote~\ref{ft}), but
can be proved in a similar way.

\begin{proposition}\label{prop:taustar-tau}
  For any rule~$R$, 
  $(\tau^*R)^{\rm prop}$ is strongly equivalent to~$\tau R$.
\end{proposition}

Since standard interpretations of~$\sigma_0$ are semi-Herbrand
(Section~\ref{sec:mainlemma}), 
the correpondence between
tuples {\bf d} of elements of domains of a standard interpretation
and tuples ${\bf d}^*$ of their names
is one-to-one, and we take the liberty to identify them. 
Therefore, for a standard interpretation~$I$ of~$\sigma_0$,
$I^\dar$ is identified with the set of precomputed atoms that are
satisfied by~$I$.  In view of this convention,
the transformation $I\mapsto I^\dar$ is the inverse of the transformation
$\J\mapsto \J^\uparrow$, defined in Section~\ref{ssec:formulas}:
for any set~$\J$ of precomputed atoms over~$\sigma_0$,
\beq
(\J^\uparrow)^\dar=\J.
\eeq{updown}

\begin{lemma}[\citeNP{fan23}, Lemma~2(i)]\label{mirek_prop2}
  A standard interpretation $I$ of~$\sigma_0$
  satisfies a sentence~$F$ over~$\sigma_0$ iff
  $I^\dar$ satisfies $F^{\rm prop}$.
\end{lemma}

The lemma below relates the meaning of a sentence in the logic of here-and-there to the meaning of its grounding in the infinitary version of that logic~\cite[Section~2]{tru12}.
It is similar to Proposition~4 from that paper and can be proved by induction in a similar way.

\begin{lemma}\label{mirek_prop4}
  A standard HT-interpretation $\langle \HH,I\rangle$ of~$\sigma_0$
  satisfies a sentence~$F$
  over~$\sigma_0$ iff $\langle \HH,I^\dar\rangle$
  satisfies~$F^{\rm prop}$.
\end{lemma}

The following lemma relates stable models of first-order formulas
(as defined in~\ref{app:ht})
to stable models of infinitary propositional formulas~\cite[Section~2]{tru12}. 
It is a generalization of Theorem~5 from that paper.

\begin{lemma}\label{mirek-th5}
  A standard interpretation~$I$ of~$\sigma_0$ is
  a stable model of a set~$\Gamma$ of sentences
  over $\sigma_0$ iff $I^\dar$ is a stable model of
    $\{F^{\rm prop}\,:\,F\in\Gamma\}$.
\end{lemma}

\begin{proof}
  By Lemma~\ref{mirek_prop2}, an interpretation $I$ is a
  model of~$\Gamma$ iff~$I^\dar$ is a model of
  \hbox{$\{F^{\rm prop}\,:\,F\in\Gamma\}$}.
  By  Lemma~\ref{mirek_prop4}, for any proper
  subset~$\HH$ of~$I^\dar$, $\langle\HH,I\rangle$ satisfies~$\Gamma$ iff
  $\langle\HH,I^\dar\rangle$ satisfies
  \hbox{$\{F^{\rm prop}\,:\,F\in\Gamma\}$}.
\end{proof}

\begin{proof}[Proof of Theorem~\ref{thm1}]
  For any set~$\J$ of precomputed atoms and any \mg\ program~$\Pi$,
  \medskip

  $\J^{\uparrow}$ is a stable model of $\tau^*\Pi$

    \quad iff

  $(\J^{\uparrow})^{\downarrow}$ is a stable model of $\{ F^{\rm
    prop} : F \in \tau^*\Pi\}$ (Lemma~\ref{mirek-th5})

    \quad iff

    $\J$ is a stable model of $\{ F^{\rm prop} : F \in \tau^*\Pi\}$
      (formula (\ref{updown}))

    \quad iff

    $\J$ is a stable model of $\{  (\tau^*R)^{\rm prop} : R \in\Pi\}$
    (definition of~$\tau^*\Pi$)

        \quad iff

    $\J$ is a stable model of $\{  \tau R : R \in\Pi\}$
    (Proposition~\ref{prop:taustar-tau})

        \quad iff

    $\J$ is a stable model of $\Pi$  (semantics of \mg).
\end{proof}

\section{Proof of Theorem~\ref{thm2}} \label{sec:proofs}

The equivalence between conditions~(b) and~(c) in the statement of
Theorem~\ref{thm2} follows from the fact that for every sentence~$F$
over~$\sigma_0$ (first-order or second-order), $\PP^\uparrow$ satisfies
$v(F)$ iff $\PP^v$ satisfies $F$.
The equivalence between conditions~(a) and~(b) will be derived from
the Main Lemma.  To this end, we need to relate the positive dependency
graph of an io-program to the positive dependency graph of the corresponding
completable set of sentences with respect to a standard interpretation
of~$\sigma_0$.  Such a
relationship is described by Lemma~\ref{lemma:local.tightness} below.

\begin{lemma}\label{lem:val.0}
  For any tuple $\bf t$ of ground terms in the language of
  mini-\gringo\ and for any tuple $\bf r$ of precomputed
  terms of the same length, the formula $\val{{\bf t}}{{\bf r}}$
  is equivalent to~$\top$ if ${\bf r} \in [{\bf t}]$,
  and to $\bot$ otherwise.
\end{lemma}

\begin{proof}
  The special case when $\bf t$ is a single term is Lemma~1 by
  Lifschitz et al.~\citeyear{lif20}.  The general case easily follows.
\end{proof}

The following lemma describes properties of the transformation $\tau^B$,
defined in Section~\ref{ssec:taustar}.

\begin{lemma}\label{lem:taub.atom}
  For any set~$\J$ of precomputed atoms and any ground literal~$L$
  such that~$\J^\uparrow \models \tau^B(L)$,
  \begin{itemize}
  \item[(a)] If~$L$ is $p({\bf t})$ or $\no\ \no\ p({\bf t})$
      then, for some tuple $\bf r$ in $[{\bf t}]$, $p({\bf r})\in\J$.
  \item[(b)]  If~$L$ is $\no\ p({\bf t})$
      then, for some tuple $\bf r$ in $[{\bf t}]$, $p({\bf r})\not\in\J$.
\end{itemize}
\end{lemma}

\begin{proof}
Consider the case when~$L$ is $p({\bf t})$.
  Then, $\tau^B(L)$ is
$\exists {\bf V}(\val{{\bf t}}{{\bf V}} \land p({\bf V}))$.
Since~$\J^\uparrow$ satisfies this formula, there exists a tuple~$\bf r$
of precomputed terms such that $\J^\uparrow$ satisfies
$\val{{\bf t}}{{\bf r}}$ and $p({\bf r})$.  Since $\val{{\bf t}}{{\bf r}}$
is satisfiable, we can conclude by Lemma~\ref{lem:val.0}
that ${\bf r}\in[{\bf t}]$.
The other two cases are analogous.
\end{proof}

\begin{lemma}\label{lem:comm}
  If $\bf U$ is a tuple of general variables, and~$\bf u$ is a tuple of
  precomputed terms of the same length, then
  \begin{itemize}
  \item[(a)] for any term $t({\bf U})$ and any precomputed term~$r$,
    the result of substituting $\bf u$ for $\bf U$ in $\val{t({\bf U})}r$
    is $\val{t({\bf u})}r$;    
  \item[(b)] for any conjunction $\body({\bf U})$ of literals and
    comparisons,
%   the result of substituting $\bf u$ for $\bf U$ in $\tau^*(\body({\bf U}))$
%   is $\tau^*(\body({\bf u}))$.
    the result of substituting $\bf u$ for $\bf U$ in $\tau^B(\body({\bf U}))$
    is $\tau^B(\body({\bf u}))$.
\end{itemize}
\end{lemma}
	
\begin{proof}
Part~(i) is easy to prove by induction.  Part~(ii) immediately follows.
\end{proof}

In the rest of this section,
$v$ is a valuation on the set of placeholders of an io-program~$\Omega$,
and~$\J$ is a set of precomputed atoms that do not
contain placeholders of~$\Omega$.  Standard interpretations of~$\sigma_0$
are semi-Herbrand, and the references to \emph{Pos} and $G^{sp}$ below
refer to the definitions modified as described at the end of
Section~\ref{sec:mainlemma}.

\begin{lemma}\label{lem:main.aux.graph2}
  Let~$\bf U$ be a list of distinct general variables, and let
  $p({\bf t}({\bf U})) \ar \body({\bf U})$ be a basic rule
  of~$\Omega$ with all variables explicitly shown.
  For any tuple~$\bf u$ of precomputed terms of the same length as~$\bf U$,
  any tuple $\bf r$ from $[v({\bf t}({\bf u}))]$, and any
  precomputed atom~$A$ from
  $\Pos{\J^\uparrow}{\tau^B(v(\body({\bf u})))}$,
  the positive dependency graph of~$\Omega$ for the input~$(v,\J^{in})$
  has an edge from $p({\bf r})$ to~$A$.
\end{lemma}

\begin{proof}
  We will show that the instance
  $$p(v({\bf t}({\bf u}))) \ar v(\body({\bf u}))$$
  of the rule
  $$p(v({\bf t}({\bf U}))) \ar v(\body({\bf U}))$$
  satisifies conditions (a)--(e) imposed on~$R$ in the definition of the
  positive dependency graph of an io-program (see Section~\ref{sec:main}).

  Verification of condition (a): the argument of~$p$ in the head of~$R$ is
  $v({\bf t}({\bf u}))$, and ${\bf r} \in [v({\bf t}({\bf u}))]$.

  Verification of condition (b):
  since $A\in\Pos{\J^\uparrow}{\tau^B(v(\body({\bf u})))}$, the
  body $v(\body({\bf u}))$ of~$R$ has a conjunctive term
  $p'({\bf r}')$ such that
  $$A\in\Pos{\J^\uparrow}{\tau^B(p'({\bf r}'))}
  =\Pos{\J^\uparrow}{\exists {\bf V}({\bf r}'={\bf V} \land p'({\bf V}))}
  =\Pos{\J^\uparrow}{p'({\bf r}'))}\subseteq\{p'({\bf r}')\}.$$
  It follows that $A=p'({\bf r}')$, so
  that~$A$ is a conjunctive term of the body
  required by condition~(b).

  Verification of conditions (c)--(e):
  let $L$ be a conjunctive term of the body
  $v(\body({\bf u}))$ of~$R$ that contains an input symbol of~$\Omega$
  or is a comparison.  Since
$\Pos{\J^\uparrow}{\tau^B(v(\body({\bf u})))}$ is non-empty,
$\tau^B(v(\body({\bf u})))$ is satisfied by $\J^\uparrow$.  Therefore, the
conjunctive term~$\tau^B(L)$ of that formula is
satisfied by $\J^\uparrow$ as well.  If~$L$ has the form $q({\bf t})$
or $\no\ \no\ q({\bf t})$ then, by Lemma~\ref{lem:taub.atom}(a), there
exists a tuple~${\bf r}'$ in $[{\bf t}]$ such that $q({\bf r}')\in\J$,
and consequently $q({\bf r}')\in\J^{in}$. If~$L$ has the form
$\no\ q({\bf t})$ then, by Lemma~\ref{lem:taub.atom}(b), there
exists a tuple~${\bf r}'$ in $[{\bf t}]$ such that $q({\bf r}')\not\in\J$,
and consequently $q({\bf r}')\notin\J^{in}$.  If~$L$ is a comparison
$t_1\prec t_2$ then $\tau^B(L)$ is
$$\exists V_1 V_2 (\val{t_1}{V_1} \land \val{t_2}{V_2} \land V_1\prec V_2).$$
Since this formula is satisfied by~$\J^\uparrow$,
the relation~$\prec$ holds for a pair $(r_1,r_2)$ of
precomputed terms such that $\val{t_1}{r_1}$ and $\val{t_2}{r_2}$
are satisfied by $\J^\uparrow$. By Lemma~\ref{lem:val.0},
$r_1\in[t_1]$ and $r_2\in[t_2]$.
\end{proof}

\begin{lemma}\label{lemma:local.tightness}
Let~$\Pi$ be the set of rules of~$\Omega$.
  \begin{itemize}
    \item[(i)]
The positive dependency graph of~$\Omega$ for the input~$(v,\J^{in})$
is a supergraph of the positive dependency graph of~$\tau^*(v(\Pi))$.
    \item[(ii)]  If~$\Omega$ is locally tight for the input~$(v,\J^{in})$
  then the graph $G^{sp}_{\J^\uparrow}(\tau^*(v(\Pi)))$ has no infinite walks.
\end{itemize}      
\end{lemma}

\begin{proof}
(i)~Replacing a choice rule $\{p({\bf t})\}\ar\body$ in~$\Omega$ by the
  basic rule $p({\bf t})\ar\body$ affects neither the positive dependency
  graph of~$\Omega$ nor the positive dependency graph of~$\tau^*(v(\Pi))$.
  Consequently we can assume, without loss of generality, that~$\Omega$ is an
  io-program without choice rules.

  Pick any edge~$(p({\bf r}),A)$ of the (modified) graph
  $G^{sp}_{\J^\uparrow}(\tau^*(v(\Pi)))$.
  Then, there is an instance $F\to G$ of the completed sentence obtained by
  applying $\tau^*$ to a basic rule of~$v(\Pi)$ such that
  \hbox{$p({\bf r})\in\Pos {\J^\uparrow}G$} and $A\in\Pos {\J^\uparrow}F$.
  That rule of~$v(\Pi)$ can be written as
  \beq
  p(v({\bf t}({\bf U})))\ar v(\body({\bf U})),
  \eeq{vrule}
  where
$$
  p({\bf t}({\bf U}))\ar \body({\bf U})
$$
is a rule of~$\Pi$, and $\bf U$ is the list of its variables.
  The result of applying~$\tau^*$ to~(\ref{vrule}) is
  $$\forall {\bf U}{\bf V}
  (\val{v({\bf t}({\bf U}))}{{\bf V}}
  \land \tau^B(v(\body({\bf U}))) \rar p({\bf V})),$$
  where $\bf V$ is a tuple of general variables.  Let~$\bf u$,~$\bf z$ be
  tuples of precomputed terms that are substituted for~$\bf U$,~$\bf V$ in
  the process of forming the instance $F\to G$.  By Lemma~\ref{lem:comm},
  that instance can be written as
  $$\val{v({\bf t}({\bf u}))}{\boldv}
  \land \tau^B(v(\body({\bf u}))) \rar p(\boldv).$$
  Consequently $p({\bf r})\in\Pos{\J^\uparrow}{p(\boldv)}$
  and
  $$A\in\Pos{\J^\uparrow}{\val{v({\bf t}({\bf u}))}{\boldv}
    \land \tau^B(v(\body({\bf u})))}.$$
  The first of these conditions implies $\boldv={\bf r}$, so that the
  second can be rewritten as
  \beq
  A\in\Pos{\J^\uparrow}{\val{v({\bf t}({\bf u}))}{{\bf r}}
    \land \tau^B(v(\body({\bf u})))}.\eeq{mm}
  It follows that $\J^{\uparrow}$ satisfies
  $\val{v({\bf t}({\bf u}))}{{\bf r}}$.  By Lemma~\ref{lem:val.0}, we
  can conclude that ${\bf r}\in v({\bf t}({\bf u}))$.  On the other hand, the
  set $\Pos{\J^\uparrow}{\val{v({\bf t}({\bf u}))}{\bf r}}$ is empty,
    because the formula $\val{v({\bf t}({\bf u}))}{{\bf r}}$ does not contain
    intensional symbols.  Consequently~(\ref{mm}) implies that
    $$A\in\Pos{\J^\uparrow}{\tau^B(v(\body({\bf u})))}.$$
    By Lemma~\ref{lem:main.aux.graph2}, it follows that
    $(p({\bf r}),A)$ is an edge of the positive dependency graph of~$\Omega$
    for the input~$(v,\J^{in})$.

    (ii)~Immediate from part~(i).
\end{proof}

Now we are ready to prove that condition~(a) in the statement of
Theorem~\ref{thm2} is equivalent to condition~(b).
The assumption that~$\PP$ is
an io-model of~$\Omega$ for an input~$(v,\I)$ implies
that~$\PP^{in}=\I$.  Indeed, that assumption means that~$\PP$ is the
set of public atoms of some stable model~$\J$ of the program $v(\Pi)\cup\I$,
where~$\Pi$ is the set of rules of~$\Omega$.  In that program, the
facts~$\I$ are the only rules containing input symbols in the head, so
that $\J^{in}=\I$.  Since all input symbols are public, it follows
that $\PP^{in}=\I$.  It remains to show that
\beq
\hbox{$\PP$ is an io-model of~$\Omega$ for~$(v,\PP^{in})$ iff 
  $\PP^\uparrow$ satisfies $v(\comp[\Omega])$}
\eeq{thm2b}
under the assumption that $\Omega$ is locally tight for~$(v,\PP^{in})$.

Recall that~$\comp[\Omega]$ is the second order sentence $\exists P_1 \cdots P_l\,C$, where~$C$ is obtained from the sentence
$\comp[\tau^*\Pi]$ by replacing the private symbols $p_1/n_1,\dots,p_l/n_l$
of~$\Omega$ with distinct predicate variables $P_1,\dots,P_l$.
Then $v(\comp[\Omega])$ is $\exists P_1 \cdots P_l\,v(C)$.
Hence, the condition in the right\nobreakdash-hand side of~(\ref{thm2b}) means that
$$\ba l
\hbox{for some set~$\J$ obtained from~$\PP$ by adding private precomputed
  atoms,}\\
\J^\uparrow \hbox{ satisfies } v(\comp[\tau^*\Pi]).
\ea$$
Note also
that~$v(\comp[\tau^*\Pi]) = \comp[\tau^*(v(\Pi)]$.  Thus the right-hand
side of~(\ref{thm2b}) is equivalent to the condition
\begin{gather}
  \begin{aligned}
    &\hbox{for some set~$\J$ obtained from~$\PP$ by adding private precomputed
    atoms,}\\
  &\J^\uparrow \hbox{ satisfies } \comp[\tau^*(v(\Pi))].
  \end{aligned}
  \label{thm2c}
\end{gather}
The condition in the left\nobreakdash-hand side of~(\ref{thm2b})
means that~$\PP$ is the set of public atoms of some stable model of
$v(\Pi) \cup \PP^{in}$. 
By Theorem~\ref{thm1}, this condition can be
equivalently reformulated as follows: 
\begin{gather}
 \begin{aligned}
   &\hbox{for some set~$\J$ obtained from~$\PP$ by adding private precomputed
   atoms,}\\
 &\hbox{$\J^\uparrow$ is a stable model of
 $\tau^*(v(\Pi) \cup \PP^{in})$.}
 \end{aligned}
  \label{thm2dd}
\end{gather}
We can further reformulate this condition
using the Main Lemma (Section~\ref{sec:mainlemma}) with
\begin{itemize}
\item the signature~$\sigma_0$ as~$\sigma$,
\item all comparison symbols viewed as extensional,
\item $\J^\uparrow$ as~$I$, and
\item $\tau^*(v(\Pi) \cup \PP^{in})$ as $\Gamma$.
\end{itemize}
The graph
$G^{sp}_{\J^\uparrow}(\tau^*(v(\Pi)\cup \PP^{in}))$
has no infinite walks, because
it is identical to the graph $G^{sp}_{\J^\uparrow}(\tau^*(v(\Pi)))$, which
has no infinite walks by Lemma~\ref{lemma:local.tightness}(ii). 
Consequently~\eqref{thm2dd}---and so the left-hand side of~(\ref{thm2b})---is equivalent to the condition:
\begin{gather}
  \begin{aligned}
    &\hbox{for some set~$\J$ obtained from~$\PP$ by adding private  precomputed atoms,}\\
    &\hbox{$\J^\uparrow$ satisfies $\comp[\tau^*(v(\Pi)) \cup \PP^{in}]$}.
  \end{aligned}
  \label{thm2e}
\end{gather}
Therefore, condition~\eqref{thm2b} is equivalent to the assertion that conditions~\eqref{thm2c} and~\eqref{thm2e} are equivalent to each other.
So to prove~\eqref{thm2b}, it is enough to show that
\begin{gather}
  \begin{aligned}
    &\hbox{for every set~$\J$ obtained from~$\PP$ by adding private precomputed atoms,}\\
    &\hbox{$\J^\uparrow$  satisfies $\comp[\tau^*(v(\Pi))]$ iff $\J^\uparrow$ satisfies $\comp[\tau^*(v(\Pi)) \cup \PP^{in}]$}.
  \end{aligned}
  \label{thm2f}
\end{gather}
Formula $\comp[\tau^*(v(\Pi)) \cup \PP^{in}]$ is a conjunction
that includes
the completed definitions of all input symbols among its
conjunctive terms.  The
interpretation~$\J^\uparrow$ satisfies these completed definitions, because
$\J^{in}=\PP^{in}$. 
On the other hand, the remaining conjunctive terms of
$\comp[\tau^*(v(\Pi)) \cup \PP^{in}]$ form the completion
$\comp[\tau^*(v(\Pi))]$ under the 
assumption that both comparison symbols and the input symbols are considered
extensional.
Hence, \eqref{thm2f} holds.

\section{Proof of Theorem~\ref{thm3}}\label{sec:proofthm3}

IO-programs~$\Omega$ and~$\Omega'$ are equivalent to each other under
an assumption~$A$ iff, for every input~$(v,\I)$ from the domain defined
by~$A$ and every set~$\PP$ of precomputed public atoms,
\begin{gather*}
  \text{$\PP$ is an io\nobreakdash-model of~$\Omega$ for~$(v,\I)$
    iff $\PP$ is an io\nobreakdash-model of~$\Omega'$ for~$(v,\I)$. }
\end{gather*}
Since $\Omega$ and~$\Omega'$ are locally tight under assumption~$A$, they both are locally tight for input~$(v,\I)$.
Hence, by Theorem~\ref{thm2}, the above condition is equivalent to:
\begin{gather*}
  \text{
    $\PP^v$ satisfies~$\comp[\Omega]$ and~$\PP^{\mathit{in}} = \I$
    \quad iff\quad
    $\PP^v$ satisfies~$\comp[\Omega']$ and~$\PP^{\mathit{in}} = \I$.
  }
\end{gather*}
We can further rewrite this condition as
\begin{gather*}
  \text{
    $\PP^{\mathit{in}} = \I$
    implies that
    $\PP^v$ satisfies~$\comp[\Omega] \leftrightarrow \comp[\Omega']$.
  }
\end{gather*}
Hence,
programs~$\Omega$ and~$\Omega'$ are equivalent to each other under assumption~$A$ iff the condition
\begin{gather}
  \text{
    $\PP^v$ satisfies~$\comp[\Omega] \leftrightarrow \comp[\Omega']$
  }
  \label{thm3.b}
\end{gather}
holds for every input~$(v,\I)$ from the domain defined by~$A$ and every
set~$\PP$ of precomputed public atoms such that~$\PP^{\mathit{in}} = \I$;
equivalently, iff
\begin{center}
\eqref{thm3.b} holds for every input~$(v,\I)$ such that~$\I^v$ satisfies~$A$\\
and every set~$\PP$ of precomputed public atoms\\
such that~$\PP^{\mathit{in}} = \I$;
\end{center}
equivalently, iff
\begin{center}
\eqref{thm3.b} holds for every input~$(v,\I)$\\
and every set~$\PP$ of precomputed public atoms\\
such that $\PP^{\mathit{in}} = \I$ and~$\PP^v$ satisfies~$A$;
\end{center}
equivalenly, iff
\begin{center}
\eqref{thm3.b} holds for every valuation~$v$ of \emph{PH}\\
and every set~$\PP$ of precomputed public atoms
such that $\PP^v$ satisfies~$A$;
\end{center}
equivalently, iff for every valuation~$v$ of \emph{PH} and every
set~$\PP$ of precomputed public atoms, $\PP^v$ satisfies the
sentence
\begin{gather*}
A \to (\comp[\Omega] \leftrightarrow \comp[\Omega']).
\end{gather*}
It remains to observe that an interpretation of~$\sigma_0$ is standard
for~$PH$ iff it can be represented in the form~$\PP^v$ for some~$\PP$ and~$v$.

\section{Conclusion}

In this paper, we defined when an io-program is considered locally tight for
an input, and proved two properties of locally tight programs: io-models
can be characterized in terms of completion~(Theorem~\ref{thm2}), and
the \anthemp\ algorithm can be used to verify equivalence (Theorem~\ref{thm3}).
The local tightness condition is satisfied for many nontight programs that
describe dynamic domains.

Theorem~\ref{thm2} can be used also to justify using the \anthem\ algorithm
\cite{fan20} for verifying a locally tight program with respect to a
specification expressed by first-order sentences.

Future work will include extending the \anthemp\ algorithm and
Theorem~\ref{thm3} to ASP programs that involve aggregate expressions,
using the ideas of Fandinno et al.~\citeyear{fan22a}
and Lifschitz \citeyear{lif22a}.

\section*{Acknowledgements} 

Many thanks to Yuliya Lierler and to
the anonymous referees for their valuable comments.

\bibliographystyle{acmtrans}
\bibliography{bib,procs}

\begin{thebibliography}{}

\bibitem[\protect\citeauthoryear{Clark}{Clark}{1978}]{cla78}
{\sc Clark, K.} 1978.
\newblock Negation as failure.
\newblock In {\em Logic and Data Bases}, {H.~Gallaire} {and} {J.~Minker}, Eds.
  Plenum Press, New York, 293--322.

\bibitem[\protect\citeauthoryear{Fages}{Fages}{1994}]{fag94}
{\sc Fages, F.} 1994.
\newblock Consistency of {C}lark's completion and existence of stable models.
\newblock {\em Journal of Methods of Logic in Computer Science\/}~{\em 1},
  51--60.

\bibitem[\protect\citeauthoryear{Fandinno, Hansen, and Lierler}{Fandinno
  et~al\mbox{.}}{2022}]{fan22a}
{\sc Fandinno, J.}, {\sc Hansen, Z.}, {\sc and} {\sc Lierler, Y.} 2022.
\newblock Axiomatization of aggregates in answer set programming.
\newblock In {\em Proceedings of the AAAI Conference on Artificial
  Intelligence}.

\bibitem[\protect\citeauthoryear{Fandinno, Hansen, Lierler, Lifschitz, and
  Temple}{Fandinno et~al\mbox{.}}{2023}]{fan23a}
{\sc Fandinno, J.}, {\sc Hansen, Z.}, {\sc Lierler, Y.}, {\sc Lifschitz, V.},
  {\sc and} {\sc Temple, N.} 2023.
\newblock External behavior of a logic program and verification of refactoring.
\newblock {\em Theory and Practice of Logic Programming\/}.

\bibitem[\protect\citeauthoryear{Fandinno and Lifschitz}{Fandinno and
  Lifschitz}{2023}]{fan23}
{\sc Fandinno, J.} {\sc and} {\sc Lifschitz, V.} 2023.
\newblock Omega-completeness of the logic of here-and-there and strong
  equivalence of logic programs.
\newblock In {\em Proceedings of International Conference on Principles of
  Knowledge Representation and Reasoning}.
\newblock To appear ({\tt
  https://www.cs.utexas.edu/users/vl/papers/omega.pdf}).

\bibitem[\protect\citeauthoryear{Fandinno, Lifschitz, L\"uhne, and
  Schaub}{Fandinno et~al\mbox{.}}{2020}]{fan20}
{\sc Fandinno, J.}, {\sc Lifschitz, V.}, {\sc L\"uhne, P.}, {\sc and} {\sc
  Schaub, T.} 2020.
\newblock Verifying tight logic programs with {A}nthem and {V}ampire.
\newblock {\em Theory and Practice of Logic Programming\/}~{\em 20}.

\bibitem[\protect\citeauthoryear{Ferraris}{Ferraris}{2005}]{fer05}
{\sc Ferraris, P.} 2005.
\newblock Answer sets for propositional theories.
\newblock In {\em Proceedings of International Conference on Logic Programming
  and Nonmonotonic Reasoning ({LPNMR})}. 119--131.

\bibitem[\protect\citeauthoryear{Ferraris, Lee, and Lifschitz}{Ferraris
  et~al\mbox{.}}{2011}]{fer09}
{\sc Ferraris, P.}, {\sc Lee, J.}, {\sc and} {\sc Lifschitz, V.} 2011.
\newblock Stable models and circumscription.
\newblock {\em Artificial Intelligence\/}~{\em 175}, 236--263.

\bibitem[\protect\citeauthoryear{Gebser, Harrison, Kaminski, Lifschitz, and
  Schaub}{Gebser et~al\mbox{.}}{2015}]{geb15}
{\sc Gebser, M.}, {\sc Harrison, A.}, {\sc Kaminski, R.}, {\sc Lifschitz, V.},
  {\sc and} {\sc Schaub, T.} 2015.
\newblock Abstract {G}ringo.
\newblock {\em Theory and Practice of Logic Programming\/}~{\em 15}, 449--463.

\bibitem[\protect\citeauthoryear{Gebser, Kaminski, Kaufmann, Lindauer,
  Ostrowski, Romero, Schaub, and Thiele}{Gebser
  et~al\mbox{.}}{2019}]{gringomanual}
{\sc Gebser, M.}, {\sc Kaminski, R.}, {\sc Kaufmann, B.}, {\sc Lindauer, M.},
  {\sc Ostrowski, M.}, {\sc Romero, J.}, {\sc Schaub, T.}, {\sc and} {\sc
  Thiele, S.} 2019.
\newblock Potassco {U}ser {G}uide.
\newblock Available at \url{https://github.com/potassco/guide/releases/}.

\bibitem[\protect\citeauthoryear{Gebser, Kaufmann, Neumann, and Schaub}{Gebser
  et~al\mbox{.}}{2007}]{geb07a}
{\sc Gebser, M.}, {\sc Kaufmann, B.}, {\sc Neumann, A.}, {\sc and} {\sc Schaub,
  T.} 2007.
\newblock {\sc clasp}: A conflict-driven answer set solver.
\newblock In {\em Proceedings of the Ninth International Conference on Logic
  Programming and Nonmonotonic Reasoning ({LPNMR}'07)}.

\bibitem[\protect\citeauthoryear{Gebser, Schaub, and Thiele}{Gebser
  et~al\mbox{.}}{2007}]{geb07b}
{\sc Gebser, M.}, {\sc Schaub, T.}, {\sc and} {\sc Thiele, S.} 2007.
\newblock Gringo: A new grounder for answer set programming.
\newblock In {\em Proceedings of the Ninth International Conference on Logic
  Programming and Nonmonotonic Reasoning}. 266--271.

\bibitem[\protect\citeauthoryear{Harrison, Lifschitz, Pearce, and
  Valverde}{Harrison et~al\mbox{.}}{2017}]{har17}
{\sc Harrison, A.}, {\sc Lifschitz, V.}, {\sc Pearce, D.}, {\sc and} {\sc
  Valverde, A.} 2017.
\newblock Infinitary equilibrium logic and strongly equivalent logic programs.
\newblock {\em Artificial Intelligence\/}~{\em 246}, 22--33.

\bibitem[\protect\citeauthoryear{Kova\'cs and Voronkov}{Kova\'cs and
  Voronkov}{2013}]{vor13}
{\sc Kova\'cs, L.} {\sc and} {\sc Voronkov, A.} 2013.
\newblock First-order theorem proving and {V}ampire.
\newblock In {\em International Conference on Computer Aided Verification}.
  1--–35.

\bibitem[\protect\citeauthoryear{Lee and Meng}{Lee and Meng}{2011}]{lee11b}
{\sc Lee, J.} {\sc and} {\sc Meng, Y.} 2011.
\newblock First-order stable model semantics and first-order loop formulas.
\newblock {\em Journal of Artificial Inteligence Research (JAIR)\/}~{\em 42},
  125--180.

\bibitem[\protect\citeauthoryear{Leone, Pfeifer, Faber, Eiter, Gottlob, Perri,
  and Scarcello}{Leone et~al\mbox{.}}{2006}]{leo06a}
{\sc Leone, N.}, {\sc Pfeifer, G.}, {\sc Faber, W.}, {\sc Eiter, T.}, {\sc
  Gottlob, G.}, {\sc Perri, S.}, {\sc and} {\sc Scarcello, F.} 2006.
\newblock The {DLV} system for knowledge representation and reasoning.
\newblock {\em ACM Transactions on Computational Logic\/}~{\em 7,\/}~3,
  499--562.

\bibitem[\protect\citeauthoryear{Lierler and Maratea}{Lierler and
  Maratea}{2004}]{lie04}
{\sc Lierler, Y.} {\sc and} {\sc Maratea, M.} 2004.
\newblock Cmodels-2: {S}{A}{T}-based answer set solver enhanced to non-tight
  programs.
\newblock In {\em Procedings of International Conference on Logic Programming
  and Nonmonotonic Reasoning ({LPNMR})}. 346--350.

\bibitem[\protect\citeauthoryear{Lifschitz}{Lifschitz}{2019}]{lif19a}
{\sc Lifschitz, V.} 2019.
\newblock {\em Answer Set Programming}.
\newblock Springer.

\bibitem[\protect\citeauthoryear{Lifschitz}{Lifschitz}{2022}]{lif22a}
{\sc Lifschitz, V.} 2022.
\newblock Strong equivalence of logic programs with counting.
\newblock {\em Theory and Practice of Logic Programming\/}~{\em 22}.

\bibitem[\protect\citeauthoryear{Lifschitz, L\"uhne, and Schaub}{Lifschitz
  et~al\mbox{.}}{2019}]{lif19}
{\sc Lifschitz, V.}, {\sc L\"uhne, P.}, {\sc and} {\sc Schaub, T.} 2019.
\newblock Verifying strong equivalence of programs in the input language of
  gringo.
\newblock In {\em Proceedings of the 15th International Conference on Logic
  Programming and Non-monotonic Reasoning}.

\bibitem[\protect\citeauthoryear{Lifschitz, L\"uhne, and Schaub}{Lifschitz
  et~al\mbox{.}}{2020}]{lif20}
{\sc Lifschitz, V.}, {\sc L\"uhne, P.}, {\sc and} {\sc Schaub, T.} 2020.
\newblock Towards verifying logic programs in the input language of clingo.
\newblock In {\em Fields of Logic and Computation {III}, Essays Dedicated to
  Yuri Gurevich on the Occasion of His 80th Birthday}. Springer, 190--209.

\bibitem[\protect\citeauthoryear{Lifschitz, Morgenstern, and
  Plaisted}{Lifschitz et~al\mbox{.}}{2008}]{lif08b}
{\sc Lifschitz, V.}, {\sc Morgenstern, L.}, {\sc and} {\sc Plaisted, D.} 2008.
\newblock Knowledge representation and classical logic.
\newblock In {\em Handbook of Knowledge Representation}, {F.~van Harmelen},
  {V.~Lifschitz}, {and} {B.~Porter}, Eds. Elsevier, 3--88.

\bibitem[\protect\citeauthoryear{Lin and Zhao}{Lin and Zhao}{2004}]{lin04}
{\sc Lin, F.} {\sc and} {\sc Zhao, Y.} 2004.
\newblock {A}{S}{S}{A}{T}: Computing answer sets of a logic program by
  {S}{A}{T} solvers.
\newblock {\em Artificial Intelligence\/}~{\em 157}, 115--137.

\bibitem[\protect\citeauthoryear{Pearce and Valverde}{Pearce and
  Valverde}{2004}]{pea04}
{\sc Pearce, D.} {\sc and} {\sc Valverde, A.} 2004.
\newblock Towards a first order equilibrium logic for nonmonotonic reasoning.
\newblock In {\em Proceedings of {E}uropean Conference on Logics in Artificial
  Intelligence ({JELIA})}. 147--160.

\bibitem[\protect\citeauthoryear{Truszczynski}{Truszczynski}{2012}]{tru12}
{\sc Truszczynski, M.} 2012.
\newblock Connecting first-order {ASP} and the logic {FO(ID)} through reducts.
\newblock In {\em Correct Reasoning: Essays on Logic-Based AI in Honor of
  Vladimir Lifschitz}, {E.~Erdem}, {J.~Lee}, {Y.~Lierler}, {and} {D.~Pearce},
  Eds. Springer, 543--559.

\end{thebibliography}

\appendix

\section{Many-Sorted Formulas} \label{app:many-sorted}

A many-sorted signature consists of symbols of three
kinds---\emph{sorts}, \emph{function constants}, and
\emph{predicate constants}.  A reflexive and transitive \emph{subsort}
relation is defined on the set of sorts.
A tuple $s_1,\dots,s_n$ ($n\geq 0$) of \emph{argument sorts} is assigned
to every function constant and to every predicate constant; in addition, a
\emph{value sort} is assigned to every function constant.
Function constants with $n=0$ are called \emph{object constants}.

We assume that for every sort, an infinite sequence of \emph{object
  variables} of that sort is chosen.  \emph{Terms} over a many-sorted
signature~$\sigma$ are defined recursively:
\begin{itemize}
\item
  object constants and object variables of a sort~$s$ are terms of sort~$s$;
\item
  if~$f$ is a function constant with argument sorts $s_1,\dots,s_n$
  ($n>0$) and value sort~$s$, and $t_1,\dots,t_n$ are terms
  such that the sort of $t_i$ is a subsort of~$s_i$ ($i=1,\dots,n$), then
  $f(t_1,\dots,t_n)$ is a term of sort~$s$.
\end{itemize}
\emph{Atomic formulas} over~$\sigma$ are
\begin{itemize}
\item
  expressions of the form $p(t_1,\dots,t_n)$, where~$p$ is a predicate
  constant and $t_1,\dots,t_n$ are terms such that their sorts are
  subsorts of the argument sorts~$s_1,\dots,s_n$ of~$p$, and
\item
  expressions of the form $t_1=t_2$, where $t_1$ and~$t_2$ are terms such that
  their sorts have a common supersort.
\end{itemize}
\emph{First-order formulas} over~$\sigma$ are formed from atomic formulas
and the 0-place connective~$\bot$ (falsity) using the binary
connectives $\land$, $\lor$, $\to$ and the quantifiers $\forall$, $\exists$.
The other connectives are treated as abbreviations: $\neg F$ stands for
$F\to\bot$ and $F\lrar G$ stands for $(F\to G)\land (G\to F)$.
\emph{First-order sentences} are first-order formulas without free variables.
The \emph{universal closure} $\uc F$ of a formula~$F$ is the
sentence~$\forall{\bf X}\,F$, where~$\bf X$ is the list of free variables
of~$F$.

An \emph{interpretation}~$I$ of a signature~$\sigma$ assigns
\begin{itemize}
  \item
  a non-empty \emph{domain} $|I|^s$ to every sort~$s$ of~$I$, so that
  $|I|^{s_1}\subseteq |I|^{s_2}$ whenever~$s_1$ is a subsort of~$s_2$,
  \item
a function~$f^I$ from $|I|^{s_1}\times\cdots\times|I|^{s_n}$ to $|I|^s$ to
every function constant~$f$ with argument sorts $s_1,\dots,s_n$ and
value sort~$s$, and
\item a Boolean-valued function~$p^I$ on
$|I|^{s_1}\times\cdots\times|I|^{s_n}$ to every predicate constant~$p$
with argument sorts $s_1,\dots,s_n$.
\end{itemize}

If~$I$ is an interpretation of a signature~$\sigma$ then
by~$\sigma^I$ we denote the signature
obtained from~$\sigma$ by adding, for every element $d$ of a domain $|I|^s$,
its \emph{name} $d^*$ as an object constant of sort~$s$.  The
interpretation~$I$ is extended to $\sigma^I$ by defining $(d^*)^I=d$.
The value $t^I$ assigned by an interpretation~$I$ of~$\sigma$ to a ground
term~$t$ over~$\sigma^I$ and the
 satisfaction relation between an interpretation of~$\sigma$ and a
sentence over~$\sigma^I$ are defined recursively, in the usual
way.  A \emph{model} of a set~$\Gamma$ of sentences is an interpretation
that satisifes all members of~$\Gamma$.

If $\boldd$ is a tuple $d_1,\dots,d_n$ of elements of domains of~$I$
then~$\boldd^*$ stands for the tuple $d_1^*,\dots,d_n^*$ of their names.
If $\boldt$ is a tuple $t_1,\dots,t_n$ of ground terms then~$\boldt^I$
stands for the tuple $t_1^I,\dots,t_n^I$ of values assigned to them
by~$I$. 

\section{Stable models of many-sorted theories}\label{app:ht}

The definition of a stable model given below is based on the first-order
logic of
here-and-there, which was introduced by Pearce and Valverde~\citeyear{pea04}
and Ferraris \emph{et al.}~\citeyear{fer09}, and then extended to many-sorted
formulas \cite{fan23}.

Consider a many-sorted signature~$\sigma$
with its predicate constants partitioned into \emph{intensional}
and \emph{extensional}.
For any interpretation~$I$ of~$\sigma$,~$I^\dar$ stands for the
set of atoms of the form $p(\boldd^*)$ with intensional~$p$ that are
satisfied by this interpretation.\footnote{In earlier publications, this
  set of atoms was denoted by $I^{int}$.  The symbol $I^\dar$ is more
  appropriate, because this operation is opposite to the operation
  $\J^\uparrow$, as discussed in Section~\ref{ssec:taustar.prop}.}

An \emph{HT\nobreakdash-interpretation} of~$\sigma$ is a pair
$\langle \HH,I\rangle$, where $I$ is an interpretation
of~$\sigma$, and $\HH$ is a subset of $I^\dar$.
(In terms of Kripke models with two sorts,~$I$ is the there-world, and~$\HH$
describes the intensional predicates in the here-world).
The satisfaction relation~$\modelsht$ between
HT\nobreakdash-interpretation $\langle \HH, I\rangle$ of~$\sigma$
and a sentence~$F$ over~$\sigma^I$ is defined recursively as follows:
\begin{itemize}
\item
  $\langle \HH, I\rangle \modelsht p(\boldt)$,
  where~$p$ is intensional, if $p((\boldt^I)^*)\in \HH$;
\item
  $\langle \HH, I\rangle \modelsht p(\boldt)$,
  where~$p$ is extensional, if $I \models p(\boldt)$;
\item
$\langle \HH, I\rangle \modelsht t_1=t_2$ if $t_1^I=t_2^I$;
\item
$\langle \HH, I\rangle \not\modelsht\bot$;
\item
$\langle \HH, I\rangle \modelsht F\land G$ if
$\langle \HH, I\rangle \modelsht F$ and
$\langle \HH, I\rangle \modelsht G$;
\item
$\langle \HH, I\rangle \modelsht F\lor G$ if
$\langle \HH, I\rangle \modelsht F$ or
$\langle \HH, I\rangle \modelsht G$;
\item
  $\langle \HH, I\rangle \modelsht F\to G$ if
  \begin{itemize}
  \item[(i)]
    $\langle \HH, I\rangle \not\modelsht F$ or $\langle \HH, I\rangle \modelsht G$,
    and
  \item[(ii)]
    $I \models F\to G$;
    \end{itemize}
\item
  $\langle \HH, I\rangle\modelsht\forall X\,F(X)$
 if $\langle \HH, I\rangle\modelsht F(d^*)$
  for each~$d\in|I|^s$, where~$s$ is the sort of~$X$;
\item
  $\langle \HH, I\rangle\modelsht\exists X\,F(X)$
 if $\langle \HH, I\rangle\modelsht F(d^*)$
  for some~$d\in|I|^s$, where~$s$ is the sort of~$X$.
\end{itemize}
If $\langle \HH, I\rangle \modelsht F$ holds, we say that $\langle \HH, I\rangle$ \emph{satisfies}~$F$ and that $\langle \HH, I\rangle$ is an \emph{HT\nobreakdash-model} of~$F$.
If two formulas have the same HT\nobreakdash-models then we say that they are \emph{HT\nobreakdash-equivalent}.

In the following proposition, we collected some properties of
this satisfaction relation that can be proved by induction.

\begin{proposition}\label{properties}
\begin{itemize}
\item[(a)]
If $\langle\HH,I \rangle\modelsht F$ then
\hbox{$I \models F$}.
\item[(b)]
For any sentence~$F$ that does not contain intensional symbols,
$\langle\HH,I \rangle\modelsht F$ iff ${I\models F}$.
\item[(c)]
For any subset~$\SSS$
of~$\HH$ such that the predicate symbols of its members do not occur in $F$,
$\langle\HH \setminus \SSS, I\rangle\modelsht F$
iff  $\langle\HH, I\rangle\modelsht F$.
\end{itemize}
\end{proposition}

About a model~$I$ of a set~$\Gamma$ of sentences
over~$\sigma^I$ we say it is \emph{stable} if, for
every proper subset~$\HH$ of~$I^\dar$, HT\nobreakdash-interpretation
$\langle \HH,I \rangle$ does not satisfy~$\Gamma$.
In application to finite sets of formulas with a single sort, this definition
of a stable model is
equivalent to the definition in terms of the operator~SM \cite{fer09}.

We say that~$I$ is \emph{pointwise stable} if, for
every element~$M$ of~$I^\dar$, $\langle I^\dar \setminus\{M\},I\rangle$
does not satisfy~$\Gamma$.

% The proposition below relates stable models as defined above to
% stable models of infinitary propositional formulas
% \cite[Section~2]{tru12}.  It is a generalization
% of Theorem~5 from that paper.

% \begin{proposition}\label{mirek-th5}
%   An interpretation~$I$ is a stable model of a set~$\Gamma$ of sentences
%   over $\sigma^I$ iff $I^\dar$ is a stable model of
%     $\{\g IF\,:\,F\in\Gamma\}$.
% \end{proposition}

% \begin{proof}
%   By Proposition~\ref{mirek_prop2}, an interpretation $I$ is a
%   model of~$\Gamma$ iff~$I^\dar$ is a model of  $\{\g IF\,:\,F\in\Gamma\}$.
%   By  Proposition~\ref{mirek_prop4}, for any proper
%   subset~$\HH$ of~$I^\dar$, $\langle\HH,I\rangle$ satisfies~$\Gamma$ iff
%     $\langle\HH,I^\dar\rangle$ satisfies $\{\g IF\,:\,F\in\Gamma\}$.
% \end{proof}

\end{document}